\documentclass[12pt]{article}

\usepackage{amsfonts}
\usepackage{amsmath}
\usepackage{amsthm}
\usepackage{euscript}
\usepackage{graphicx}
\usepackage[utf8]{inputenc}
\usepackage{subcaption}
\usepackage{subfiles}
\usepackage{wrapfig}
\usepackage{xspace}

\newcommand{\reals}{\mathbb{R}}  

\newcommand{\PSLG}{\mathrm{PSLG}}
\newcommand{\dir}[1]{\overrightarrow{#1}} 
\newcommand{\wt}{w}

\newcommand{\grid}[1]{G_{#1}}
\newcommand{\adjgraph}[1]{\EuScript{G}_{#1}}
\newcommand{\ballgraph}[1]{\EuScript{D}_{#1}}
\newcommand{\leveledges}[1]{\EuScript{S}_{#1}} 
\newcommand{\triangulation}{\EuScript{T}}
\newcommand{\candidate}{\EuScript{A}}
\newcommand{\opt}{\EuScript{T}^*}
\newcommand{\optdiff}[1]{\EuScript{K}_{#1}}
\newcommand{\newopt}{\stackrel{\sim}{\smash{\opt_i}\rule{0pt}{1.05ex}}}
\newcommand{\newedges}[1]{\EuScript{W}_{#1}}

\newcommand{\rand}{\gamma}
\newcommand{\vertseq}{\sigma}
\newcommand{\chain}{\mathbb{C}}
\newcommand{\edges}{\EuScript{E}}
\newcommand{\cs}{\EuScript{CS}} 

\newcommand{\reflex}{\text{reflex}\xspace}
\newcommand{\convex}{\text{convex}\xspace}

\DeclareMathOperator{\fwd}{fwd}
\DeclareMathOperator{\back}{back}

\newcommand{\N}{\EuScript{N}} 

\newcommand{\bigO}[1]{\EuScript{O}(#1)}

\newtheorem{theorem}{Theorem}
\newtheorem{lemma}[theorem]{Lemma}
\newtheorem{corollary}[theorem]{Corollary}
\newtheorem{invariant}{Invariant}

\begin{document}

\begin{titlepage}

\title{A Grid-Based Approximation Algorithm for the Minimum Weight Triangulation Problem\thanks{This work is supported in part by the National Science Foundation under grant NSF-CCF 1464276. Any opinions, findings, conclusions, or recommendations expressed in this material are those of the authors and do not necessarily reflect the views of the National Science Foundation.}}
  \author{Sharath Raghvendra
      \thanks{Department of Computer Science, Virginia Tech,
      \texttt{sharathr@vt.edu}},
    Mariëtte C. Wessels
      \thanks{Department of Mathematics, Virginia Tech,
      \texttt{wmari11@vt.edu}}}
\date{}
\maketitle

\begin{abstract}
Given a set of $n$ points on a plane, in the \emph{Minimum Weight Triangulation} problem, we wish to  find a triangulation that minimizes the sum of Euclidean length of its edges. This incredibly challenging problem has been studied for more than four decades and has been only recently shown to be NP-Hard. In this paper we present a novel polynomial-time algorithm that computes a $14$-approximation of the minimum weight triangulation---a constant that is significantly smaller than what has been previously known. 

In our algorithm, we use grids to partition the edges into levels where shorter edges appear at smaller levels and edges with similar lengths appear at the same level. We then triangulate the point set incrementally by introducing edges in increasing order of their levels. We introduce the edges of any level $i+1$ in two steps. In the first step, we add edges using a variant of the well-known ring heuristic to generate a partial triangulation\  $\hat{\candidate}_i$.  In the second step, we greedily add non-intersecting level $i+1$ edges to $\hat{\candidate}_i$ in increasing order of their length and obtain a partial triangulation $\candidate_{ i+1}$. The ring heuristic is known to yield only an $\bigO{\log n}$-approximation even for a convex polygon and the greedy heuristic achieves only a $\Theta(\sqrt{n})$-approximation. Therefore, it is surprising that their combination leads to an   improved approximation ratio of $14$.

For the proof, we identify several useful properties of  $\hat{\candidate}_i$ and combine it with a new Euler characteristic based technique to show that  $\hat{\candidate}_i$ has  more edges than  $\opt_i$; here $\opt_i$ is the partial triangulation consisting of level $\le i$ edges of some minimum weight triangulation. We then use a simple greedy stays ahead proof strategy to bound the approximation ratio.    
\end{abstract}

\end{titlepage}



\section{Introduction}
\label{sec:intro}

Consider a planar point set $P \subset \reals^2$ with $|P| = n$. A triangulation $\EuScript{T}$ of $P$ is a subdivision of the interior of the convex hull of $P$ into triangles by non-intersecting straight-line segments. Alternatively, any triangulation can be viewed as a maximal planar straight-line graph ($\PSLG$) of the planar point set $P$. For any pair of points $u,v \in P$, we define the \emph{weight} of the edge $uv$ as the Euclidean distance $\|uv\|$ between $u$ and $v$.\ We define the weight $w(\EuScript{T})$ of any triangulation $\EuScript{T}$ to be the sum of the weights of all the edges in $\EuScript{T}$. The \emph{Minimum Weight Triangulation} (MWT) problem seeks to find a triangulation $\opt$ of $P$ with the smallest possible weight. Note that the MWT (also referred to as the \emph{optimal} triangulation) for a given point set is not necessarily unique. For $\alpha >1$, we define an $\alpha$-approximate minimum weight triangulation to be any triangulation $\candidate$ such that $\wt(\candidate) \le \alpha\wt(\opt)$.  In this paper, we describe a new algorithm that computes a $14$-approximate MWT  of $P$.

\subsection{Overview and Previous Work}

The problem of computing a triangulation for a given point arises naturally in applications ranging from computer graphics and cartography to finite element meshes and spatial data analysis. Different applications call for different notions of optimality, many of which has been studied and were surveyed by Bern and Eppstein \cite{BerEpp-CEG-92}. One such notion of optimality is expressed by the MWT problem. The origin of this problem dates back to 1970 in cartography where it was first considered by Düppe and Gottschalk \cite{DupGott-70} who originally proposed a greedy approach to produce a MWT. Shamos and Hoey \cite{ShamosHoey-75} conjectured that the well-known Delaunay triangulation, the dual of the Voronoi diagram, which simultaneously optimizes several objective functions might be a MWT. However, Lloyd \cite{Lloyd-77} provided examples in 1977 which show that neither the greedy nor the Delaunay triangulation is a MWT. At this point the hardness of this problem was unknown. This problem became notorious when Garey and Johnson \cite{GareyJohnson-79} included the MWT problem in their famous list of twelve major problems with unknown complexity status in 1979. It was not until 2006 when Mulzer and Rote \cite{MulRote-06} proved that the MWT problem is $NP$-hard. 

Despite the complexity of the MWT problem remaining unknown for nearly three decades, there were several efforts to design approximation algorithms. The greedy triangulation and Delaunay triangulation were shown to be a factor $\Theta(\sqrt{n})$ \cite{ManZob-79}, and a factor $\Omega(n)$ \cite{LevKrz-98, ManZob-79} approximation respectively. Plaisted and Hong \cite{PlaHong-87} described how to approximate the MWT  within a factor of $\bigO{\log{n}}$ in $\bigO{n^2 \log{n}}$ time  by partitioning the point set into empty convex polygons and then repeatedly connecting all pairs of adjacent even numbered vertices. This procedure for the convex polygon is known as the \emph{ring heuristic}. 
 Clarkson \cite{Clarkson-91} extended the ring heuristic to non-convex polygons to design an algorithm for the closely related \emph{Minimum Weight Steiner Triangulation} (MWST) problem. Levcopoulos and Krznaric \cite{LevKrz-98} showed that a variation of the greedy algorithm approximates the MWT within a very large constant. This constant was reduced by Yousefi and Young \cite{YoYo-11}, who also discuss the relationship between the MWT and integer linear programs; a relationship first noted by Dantzig \textit{et al.} \cite{Dantzig-85}. 

About two decades ago, in a major breakthrough, Arora \cite{Arora-98} introduced a shifted quadtree-based approach to compute a $(1+\varepsilon)$-approximation for the Euclidean TSP problem. Concurrent to this, but independently, Mitchell \cite{Mitch-96} also introduced a PTAS for the TSP problem. Arora's technique extended to several other geometric optimization problems. However, as noted by him, the MWT problem resists this approach. The Minimum Weight Stiener Triangulation (MWST) problem seemed amenable to Arora's technique, however, no such $(1 + \varepsilon)$-approximation algorithm has yet been found for the MWST either. Finally, Remy and Steger \cite{RemySteg-09} have discovered a quasi-polynomial approximation scheme where, for any fixed $\varepsilon$, it yields a $(1 + \varepsilon)$-approximation of the MWT in $n^{\bigO{\log^8{n}}}$ time.

A natural variant of the MWT problem is the problem of computing a triangulation that minimizes the $q$th norm of its edge costs, i.e., for an integer $q \ge 1$ compute a triangulation $\opt_q$ that minimizes the cost $\left(\sum_{uv \in \opt_q} \|uv\|^q\right)^{1/q}$. We refer to this as the $q$-MWT problem.   When $q=1$, this problem reduces to the minimum weight triangulation and when $q=\infty$, the problem reduces to the well-known minimax length triangulation problem~\cite{TanEde-93}.

\subsection{Our Approach and Results}

We present a novel polynomial time algorithm that computes a $14$-approximate MWT.  In comparison, for the minimum weight triangulation problem, the previous methods for computing a constant approximation \cite{LevKrz-98, YoYo-11} achieve an approximation ratio estimated to be higher than $3000$. We utilize a grid-based approach \cite{Eppstein-94, RemySteg-09} combined with a variant of the ring heuristic \cite{PlaHong-87,Clarkson-91} and a greedy approach \cite{LevKrz-98} to compute our approximate triangulation. 
  More specifically, we maintain a sequence of nested grids (similar to a quad-tree) $G_1,\ldots,G_k$ for some $k$. Let $G_1$ be the finest grid and $G_k$ be a single square that contains all input points. Using these grids, we partition all the $\Theta(n^2)$ edges into levels---an edge appears in level $i$ if and only if $i$ is the smallest integer such that the two end points  of this edge are in neighboring cells in $G_i$ (note that longer edges have a higher level). 

 We then triangulate the point set incrementally by introducing edges in increasing order of their levels. Our algorithm will maintain a maximal $\PSLG$ $\candidate_i$ after processing the level $i$ edges. We introduce the   level $i+1$ edges in two steps. In the first step, we add edges  and triangulate the region between  every consecutive pair of reflex chains that appear on the boundary of any non-triangulated face of $\candidate_i$. We do so by using a variant of the  ring heuristic. Let $\hat{\candidate}_i$ be the $\PSLG$ maintained at the end of this step. In the second step, we will greedily add non-intersecting level $i+1$ edges to $\hat{\candidate}_i$ in increasing order of their length to obtain a maximal $\PSLG$ $\candidate_{ i+1}$. Our algorithm returns $\candidate_k$ as the approximate triangulation. The ring heuristic is known to yield only an $\bigO{\log n}$-approximation even for a convex polygon and the greedy heuristic achieves only a $\Theta(\sqrt{n})$-approximation. Therefore, it is surprising that by combining the two heuristics, we obtain a significantly   improved $14$-approximation for this problem.

For the proof, we identify several useful properties of  $\hat{\candidate}_i$ and combine it with a new Euler characteristic based technique to show that  $\hat{\candidate}_i$ has  more edges than  $\opt_i$,  where $\opt_i$ is the restricted optimal triangulation consisting only those edges of the minimum weight triangulation which have a level of at most $i$. The approximation ratio can then be bounded by using a simple greedy stays ahead proof.

 There are three major technical contributions that assist us in the proof: 
\begin{itemize}
  \item We identify several important properties of the maximal $\PSLG$ $\candidate_i$ and the intermediate $\PSLG$  $\hat{\candidate}_i$ for each level $i$. Using these properties,  we show that the partial triangulation  $\hat{\candidate}_i$ triangulates a larger region than the restricted optimal triangulation $\opt_i$. However, this triangulated region for $\opt_i$ is not necessarily contained inside the triangulated region for $ \hat{\candidate}_i$  and this makes it difficult to compare their cardinalities. 
  \item Nevertheless, we develop a  novel technique to show that there are more  edges in $\hat{\candidate}_i$   than in $\opt_i$. This technique involves adding edges (not necessarily straight-line) to $\opt_i$ in such a way that for every non-triangulated face of $\hat{\candidate}_i$, there is a unique non-triangulated face with a greater number of edges in this augmented restricted optimal triangulation. Using Euler's formula, we can then relate the cardinality of our candidate solution and the optimal to achieve a $21$-approximation algorithm.
  \item We choose the side-length of the cells of the grid randomly and this helps to improve the expected approximation ratio to $14$.         
\end{itemize}

The analysis of ring heuristic~\cite{PlaHong-87} and the quasi-greedy algorithm~\cite{LevKrz-98} for the MWT\ problem  rely on triangle inequality of Euclidean costs.  Consequently, we cannot extend the analysis of these heuristics to the $q$-MWT problem where edge costs are $q$th powers of the Euclidean costs. Our analysis, however, does not depend on the triangle inequality and therefore easily extends to any $q$-MWT. We show that the triangulation produced by our algorithm is  a $14$-approximate $q$-MWT for every value of $q \ge 1$ including the minimum weight triangulation ($t=1)$ and the minimax length triangulation ($t=\infty$).

In the rest of the paper, we present our algorithm for the MWT problem and provide a weaker analysis of $24$- (worst-case) and $16$- (expected) approximation ratio. In Section~\ref{sec:prelim} we present the preliminary definitions  that are necessary to present our algorithm. The algorithm is described in Section~\ref{sec:algorithm}. We use our algorithmic invariants to bound the approximation ratio by a factor of $24$ (worst-case) and $16$ (expected) in Section~\ref{sec:approx-ratio}. We prove the invariants in Section~\ref{sec:invariants} and ~\ref{sec:inv2}.  In Section~\ref{sec:extension}, we will describe   an improvement of the approximation factor to $21$ (worst-case) and $14$ (expected) and also describe the extension of our analysis  to the $q$-MWT problem. We conclude in Section~\ref{sec:conclusion}.


\section{Preliminaries}
\label{sec:prelim}

Let $P \in \reals^2$ be the set of $n$ input points. For simplicity of presentation, we will assume that $P$ has a \emph{spread} of $\Delta$, i.e., the ratio of the diameter of $P$ and closest pair of points in $P$ is bounded by $\Delta$, where $\Delta$ is a power of three. We also scale and translate $P$ so that the closest pair of points in $P$ are at a distance $1$, the diameter of $P$ is bounded by $\Delta$, and all points of $P$ are enclosed inside an axis parallel $3\Delta\times 3\Delta$ square $S$ with $(0,0) $ and $(3\Delta,3\Delta)$ being the diagonally opposite corners of $S$. Note that the translation and scaling does not affect the optimal triangulation. We also assume that the points in $P$ are in general position and therefore no three points in $P$ are co-linear.  Our algorithm extends to any arbitrary point set that is not in general position and that does not have a bounded spread. However, making these assumptions simplifies the presentation of our algorithm significantly. In Section~\ref{sec:algorithm}, we will show that the running time of our algorithm is polynomial in $n$ and not dependent on the actual value of  $\Delta$.

For any pair of points, $u,v \in \reals^2$, let $\overline{uv}$ be the open straight-line segment in $\reals^2$ connecting the points $u$ and $v$ (so the endpoints $u$ and $v$ are not included in $\overline{uv}$). Consider an arbitrary graph $\mathcal{G}$ with the set $P$ as its vertex set. We denote an edge between two vertices in $\mathcal{G}$, $u, v \in P$, as $uv$. Unless otherwise noted, let $\mathcal{G}$ denote the set of edges of the graph $\mathcal{G}$, so that $|\mathcal{G}|$ denotes the number of edges in $\mathcal{G}$. The edges $uv$ and $xy$ \emph{intersect} if $\overline{uv} \cap \overline{xy }\neq \emptyset $. The graph $\mathcal{G}$ is a \emph{planar straight-line graph} ($\PSLG$) if, for any two edges $uv$ and $xy$ in $\mathcal{G}$, $uv$ and $xy$ do not intersect. For any graph $\mathcal{G}$ (not necessarily planar), a subgraph $\mathcal{M}$ is a \emph{maximal} $\PSLG$, if $\mathcal{M}$ is a $\PSLG$ and for every $uv \in \mathcal{G} \setminus \mathcal{M}$, there exists $u'v' \in \mathcal{M}$ such that $uv$ intersects $u'v'$. It is well-known that any maximal $\PSLG$ of the complete graph on $P$ is a triangulation of $P$ with $3n-3-h$ edges, where $h$ is the number of points of $P$ that appear on its convex hull. A \emph{maximum} $\PSLG$ $\mathcal{M}^* \subset \mathcal{G}$ is a maximal $\PSLG$ with the largest number of edges out of all possible maximal $\PSLG$s of $\mathcal{G}$. Note that in general, neither the maximum nor the maximal $\PSLG$ is unique.

\subsection{Grids and Adjacency Graphs}

Our algorithm is described on graphs that are induced by a sequence of nested grids constructed as follows.
First, choose a $\rand \in \reals$ uniformly at random from the open interval $(\frac{1}{3},1)$. All points in $P$ lie inside the square with diagonal corners $(0,0)$ and $(9\rand\Delta,9\rand\Delta)$. Next, define a sequence of \emph{grids} $\grid{{\log_3{9\Delta}+1}}, \ldots, \grid{0}$, where $\grid{\log_3 9\Delta+1}$ is the square $S$.  Given a grid $\grid{i+1}$, we obtain grid $\grid{i} $ by simply splitting every \emph{cell} (square) of the grid into $3\times3$ equal cells. By this construction, grid $\grid{i}$ will have $3^{\log_3 9\Delta-i+1}\times 3^{\log_3 9\Delta-i+1}$ cells each having a side-length of $\rand 3^{i-1}$. We refer to $i$ as the \emph{level} of $\grid{i}$. Without loss of generality, we may assume that, for each $\grid{i}$, every point $p \in P$ is contained in exactly one cell $C$ of $\grid{i}$. If this is not the case, then one can translate the input point set  such that no input point lies on a grid line. Clearly this does not affect the minimum weight triangulation. It is always possible to do so since $P$ consists of finitely many points. The finest grid $\grid{0}$ consists of cells with side length $\frac{\rand}{3} < \frac{1}{3}$.

For every cell $C \in \grid{i}$, we will define the \emph{neighboring cells} $N(C)$ of $C$ to be the set of cells in $\grid{i}$ that share their boundary with $C$. Our convention is to include $C$ in $N(C)$ since $C$ shares a boundary with itself (thus $N(C)$ contains at most nine cells). Let $\N(C)$ represent the geometric region formed by taking the union of all cells in the neighborhood of $C$, i.e., $\N(C) =\bigcup_{C\in N(C)} C$. For notational convenience, for any point $p \in P$ contained in the cell $C$, define $N(p) = N(C)$ and $\N(p) = \N(C)$. We refer to two cells $C$ and $C'$ of a grid $\grid{i}$ as \emph{diagonal neighbors} if their boundaries share exactly one vertex but do not share any edge. If $C$ and $C'$ share exactly one edge, they are referred to as \emph{orthogonal neighbors}. Any cell can have at most four diagonal neighbors, and four orthogonal neighbors.

For each grid $\grid{i}$, the \emph{adjacency graph} $\adjgraph{i}$ connects every pair of distinct points that are in neighboring cells with an edge. Therefore, the vertex set of the adjacency graph $\adjgraph{i}$ is $P$ itself, and for any two points $u, v \in P$, such that $u \neq v$, $uv$ is an edge in $\adjgraph{i}$ if and only if $u \in \N(v)$. It is not difficult to see that $\adjgraph{0} = \emptyset$, $\adjgraph{{\log_3{9\Delta}+1}}$ is the complete graph on $n$ vertices, and $\adjgraph{i} \subseteq \adjgraph{i + 1}$ for all $0 \leq i < \log_3{9\Delta} + 1$. As we go from level $i-1$ to level $i$ in the adjacency graph, several new edges may be added. We refer to these edges as \emph{level $i$ edges}, denoted $\leveledges{i} = \adjgraph{i} \setminus \adjgraph{i-1}$. Note that for any edge $ab \in \leveledges{i}$, it follows from the definition of a neighborhood that $\rand 3^{i-2} \leq \|ab\| \leq 6\sqrt{2}(\rand 3^{i-2})$.

\subsection{Non-triangulated faces}
A straight-line embedding of a planar graph $\mathcal{G}$, defined on the point set $P$, subdivides the plane into regions. Each region is referred to as a \emph{face} of $\mathcal{G}$. More formally, each face is a maximal connected region in $\reals^2\setminus (P \cup (\bigcup_{uv \in \mathcal{G}}\overline{uv}))$. Since a face is a subset of $\reals^2$, the \emph{boundary} of the face can be defined in the standard way (a point $x$ is on the boundary of a face $f$ if, for every $\varepsilon > 0$, the Euclidean ball centered at $x$ with radius $\varepsilon$ contains both a point in $f$ and a point not in $f$). Therefore, faces are open subsets of $\reals^2$ and their boundaries consist of points in $P$ and line segments (which corresponds to edges) of the graph $\mathcal{G}$. We refer to a bounded face $f$ as a \emph{triangulated face} if its boundary is a single connected component with exactly three edges and three vertices of $\mathcal{G}$. Any other face (bounded or unbounded) is \emph{non-triangulated}. The boundary of a non-triangulated face may consist of one or more isolated vertices, connected cycles, or trees. For example, in Figure~\ref{subfig:repeated-vertex} the unbounded face $f$ consists of two polygons and an edge on its boundary.

Let $f$ be a non-triangulated face of $\mathcal{G}$. Our algorithm applies a grid-based modified ring heuristic on a clockwise ordering of the vertices for each connected component of the boundary of $f$. We generate this clockwise ordering as a sequence $\vertseq(f)$ of the vertices that appear on the boundary of $f$ by walking along the boundary so that the face $f$ appears on the right. Let $\dir{uv}$ denote walking along any edge $uv$ from $u$ to $v$. During this construction, any edge $\dir{uv}$ has been explored if  we have already walked from $u$ to $v$ along the edge $uv$. If $uv$ has the same face on both its sides, then we will explore $\dir{uv}$ and $\dir{vu}$ as we generate the sequence for that face, which will result in a single vertex appearing multiple times in the boundary vertex sequence as part of distinct directed ``explorations". Otherwise, the edge is explored only once for this sequence (with an orientation so that $f$ appears on the right).
  We construct the vertex sequence $\vertseq(f)$ as follows: start the sequence with an arbitrary vertex $v_0$ on the boundary. If the connected component is an isolated vertex, then $\vertseq(f) =\langle v_0\rangle$ is the vertex sequence. Otherwise, let $x$ be any vertex adjacent to $v_0$ that also appears on the boundary of $f$. If $f$ appears on the right as we walk along any edge $\dir{v_0x}$, set $v_1 \leftarrow x$, the second vertex in the sequence. To determine the rest of the sequence, suppose we have generated the first $i$ vertices $\vertseq(f)=\langle v_0,\ldots,v_{i-1}\rangle$ of the sequence. To generate the $(i+1)^{\text{th}}$ vertex  $v_i$ in the sequence, we choose the potential next vertex as follows:
\begin{itemize}
  \item if $v_{i-1}$ is a degree $1$ vertex in $\mathcal{G}$, let $x$ be the only adjacent vertex;
  \item if $v_{i-1}$ had degree greater than $1$, consider the ray $r$ passing through $v_{i-2}$ and $v_{i-1}$. Let $x$ be the first vertex adjacent to $v_{i-1}$ that $r$ intersects when $r$ is rotated about $v_{i-1}$ in the anti-clockwise direction.
\end{itemize}
If $\dir{v_{i-1}x}$ has not been explored, we walk along $\overline{v_{i-1}x}$ from $v_{i-1}$ to $x$, set $v_{i}\leftarrow x$, and add $v_i$ to $\vertseq(f)$. We repeat the same process to find the next vertex in the sequence. Otherwise, if $\dir{v_{i-1}x}$ has already been explored, we return $\vertseq(f)$ as the clockwise ordering of the vertices. By repeating this process for every connected component of the boundary, we generate a separate clockwise ordering of vertices for each connected component of the boundary of $f$ (for instance, two boundary vertex sequences are generated for $f_0$ in Figure~\ref{subfig:face-boundary}).
Note that any vertex may appear multiple times in $\vertseq(f)$. In Figure~\ref{subfig:repeated-vertex}, $x_1$ appears multiple times in $\vertseq(f) = \langle v_0, v_1, \ldots, v_{10}\rangle$ but labeled as $v_1$, $v_4$, and $v_{10}$ depending on where in the sequence the vertex was encountered. In each of these instances, $x_1$ appears between two different vertices in the clockwise ordering of the boundary, and therefore generates a distinct element in $\vertseq(f)$. For $\vertseq(f)=\langle v_0,\ldots, v_{m-1}\rangle$, and for any two integers $0\le i, k\le m-1$, let $v_{i+k}$ denote the vertex $v_j$ where $j= i+k~(\mathrm{mod}~m)$. Note that one can also define a boundary vertex sequence consisting of exactly three vertices for any triangulated face of $\mathcal{G}$.
\begin{figure}
  \centering
    \begin{subfigure}{0.35\textwidth}
      \centering
      \includegraphics[width=\textwidth]{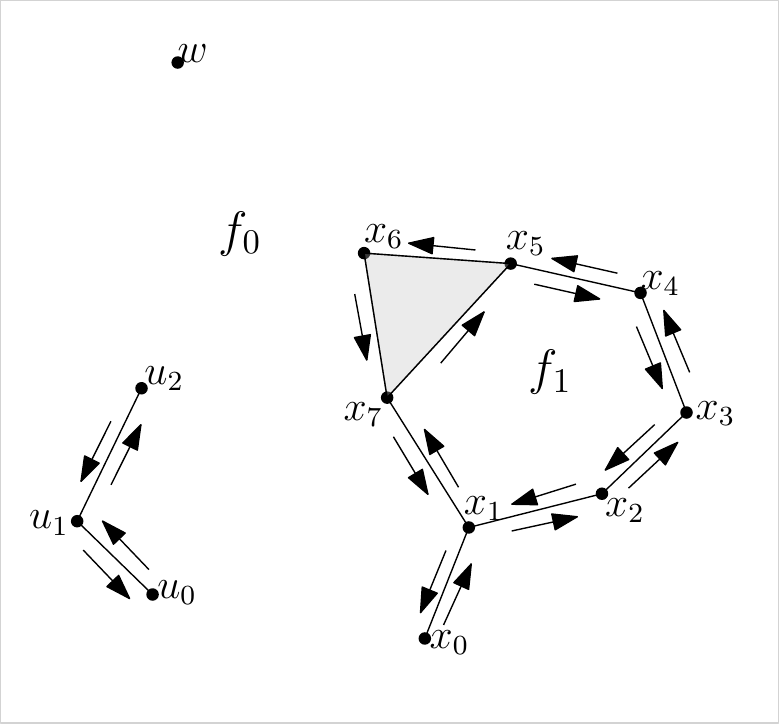}
      \caption{}
      \label{subfig:face-boundary}
    \end{subfigure}
    \hspace{8mm}
    \begin{subfigure}{0.30\textwidth}
      \centering
      \includegraphics[width=\textwidth]{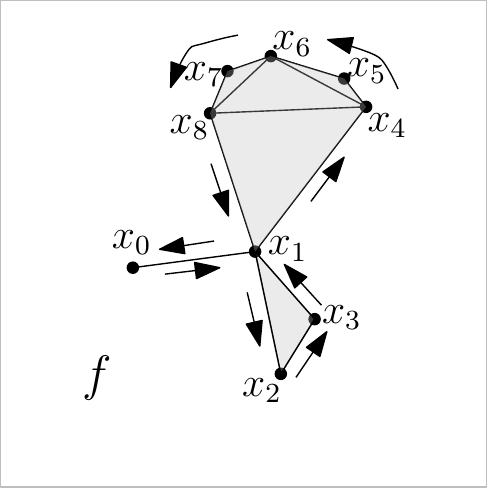}
      \caption{}
      \label{subfig:repeated-vertex}
    \end{subfigure}%
  \caption[Boundary vertex sequence $\vertseq(f)$.]{(a) $f_0$ has a disconnected boundary, and will therefore have more than one vertex sequence $\vertseq_u(f_0) = \langle u_0, u_1, u_2, u_1 \rangle$, and $\vertseq_v(f_0) = \langle x_0, x_1, x_2, x_3, x_4, x_5, x_6, x_7, x_1\rangle$. The other non-triangulated face $f_1$ has only one vertex sequence $\vertseq(f_1) = \langle x_1, x_7, x_5, x_4, x_3, x_2 \rangle$. (b) the vertex sequence constructed for $f$ is $\vertseq(f) = \langle x_0, x_1, x_2, x_3, x_1, x_4, x_5, x_6, x_7, x_8, x_1 \rangle$. Note that shaded faces are triangulated.}\label{fig:face-boundary}
\end{figure}

Suppose $\vertseq(f) = \langle v_0, \ldots, v_{m-1} \rangle$ is a vertex sequence for a non-triangulated face $f$ of $\mathcal{G}$, and let $v_i$ be any vertex in the sequence. We say that $v_i$ is a \emph{\convex vertex} if, as you walk from $v_{i-1}$ through $v_i$ to $v_{i+1}$, we make a right turn at $v_i$. Otherwise, if we make a left turn, we refer to $v_i$ as a \emph{\reflex vertex} (since we assume no three points are co-linear, every vertex is either \reflex or \convex).
For any $0\le i < j \le m-1$, we refer to the contiguous subsequence $\langle v_i,\ldots, v_j\rangle$ as a \emph{chain} from $v_i$ to $v_j$ and denote it by $\chain(v_i,v_j)$ with $|\chain(v_i,v_j)|=j-i+1$ denoting the number of vertices in the chain. The \emph{interior} of the chain is the set of all vertices in the chain except the first and last vertices $v_i$ and $v_j$. We refer to $v_i$ and $v_j$ as the boundary of the chain.  A \emph{\reflex chain} is a chain consisting only of \reflex vertices.  We refer to the \reflex chain from $v_i$ to $v_j$ as a \emph{maximal} reflex chain if $v_{i-1}$ and $v_{j+1}$ are not \reflex.

Let $\mathcal{G}$ be a $\PSLG$, and $f$ be a non-triangulated face of $\mathcal{G}$ with the boundary vertex sequence $\vertseq(f) = \langle v_0, \ldots, v_i, \ldots, v_j, \ldots, v_{m-1} \rangle$. The vertex $v_j$ is \emph{visible} to a vertex $v_i$ in $\vertseq(f)$ if $\overline{v_iv_j} \subset f$ and adding the edge $v_iv_j$ to $\mathcal{G}$ creates two faces $f'$ and $f''$ with $\vertseq(f') = \langle v_i, v_{i+1}\ldots, v_{j-1}, v_j\rangle$ and $\vertseq(f'')=\langle v_i, v_j, v_{j+1},\ldots, v_{i-1}\rangle$. Note that the definition of visibility is symmetric, i.e., for $x,y \in \sigma(f)$, $x$ is visible to $y$ if and only if $y$ is visible to $x$. We say that $v_i$ and $v_j$ are $\delta$-visible if $v_i$ is visible to $v_j$ and the length of $\overline{v_iv_j}$ is at most $\delta$.
We extend this definition and define the visibility between  a vertex and any point on a line  segment. For any $v_i \in \vertseq(f)$,  we say that a point $q \in \overline{v_jv_{j+1}} $ is visible to a vertex $v_i$ if $v_i$ is visible to $q$ in the vertex sequence $\vertseq(f) = \langle v_0,\ldots,v_{i},\ldots, v_j, q, v_{j+1},\ldots, v_{m-1}\rangle $. We say that a vertex $v_i$ is visible to the edge $v_jv_{j+1}$ if there exists a point $q \in \overline{v_jv_{j+1}}$ with the property that $v_i$ is visible to $q$. In addition, if the length of $\overline{v_iq}$ is at most $\delta$, we say that $v_jv_{j+1}$ is $\delta$-visible to $v_i$.

Let $\vertseq(f) = \langle v_0, v_1, \ldots v_{m-1} \rangle$ be the boundary vertex sequence for some non-triangulated face in a $\PSLG$. For any vertex $v_j \in \vertseq(f)$, let $k$ be the smallest integer such that $ k > j$ and $v_{k(\mathrm{mod}~m)}$ is convex. We define $v_k$ to be the \emph{forward convex vertex} of $v_j$. Also, define $v_{k+1}$ to be the \emph{forward support vertex} of $v_j$ denoted by $\fwd(v_j, \vertseq(f))$.  For any reflex vertex $v_j$  in $\vertseq(f)$, let $k$ be the largest integer such that $ k < j$ and $v_{k (\mathrm{mod}~m)}$ is convex. We define $v_k$ to be the \emph{backward convex vertex} of $v_j$. Then we define $v_{k-1}$ to be the \emph{backward support vertex} of $v_j$ denoted by $\back(v_j, \vertseq(f))$. Note that the backward convex vertex and backward support vertex is defined only for reflex vertices of $\vertseq(f)$.

\begin{figure}
  \centering
  \begin{subfigure}{0.35\textwidth}
    \centering
    \includegraphics[width=0.80\textwidth]{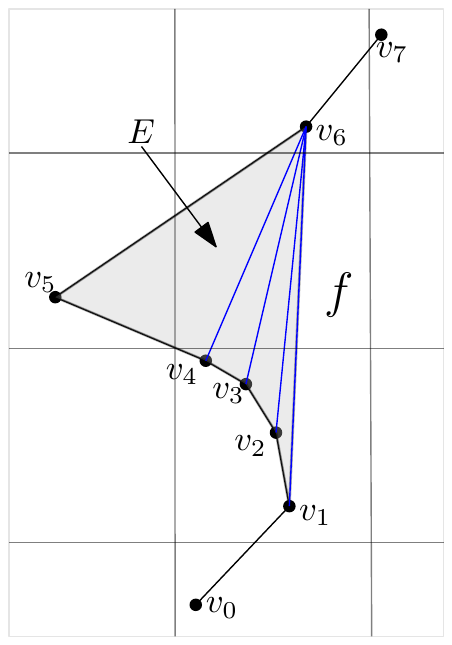}
    \caption{}\label{subfig:one-chain}
  \end{subfigure}\hspace{8mm}
  \begin{subfigure}{0.30\textwidth}
    \centering
    \includegraphics[width=0.80\textwidth]{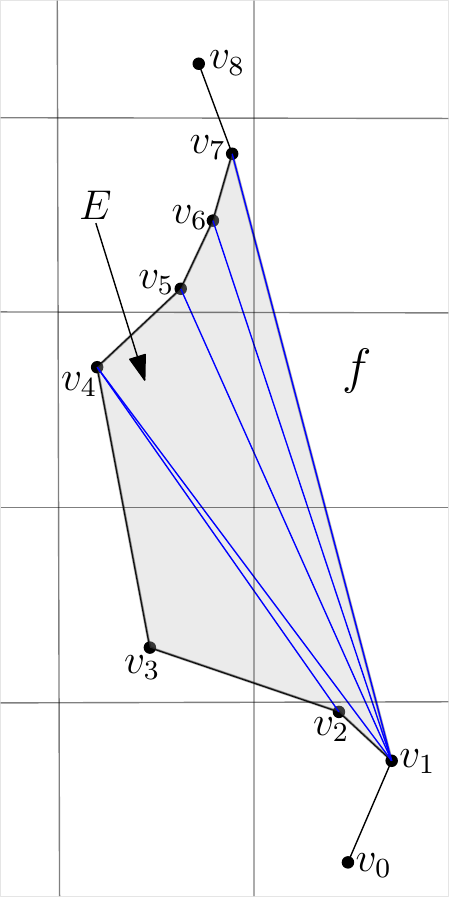}
    \caption{}\label{subfig:two-chain}
  \end{subfigure}
  \caption[Type $1$-chain and type $2$-chain.]{(a) $\chain(v_1, v_6)$ is a $1$-chain, and (b) $\chain(v_1, v_7)$ is a $2$-chain, that can be triangulated with blue edges.}\label{fig:one-two-chains}
\end{figure}
\paragraph{Non-triangulated face and the grid.} We will show that in the first part of the algorithm, while executing the modified ring heuristic on any non-triangulated face $f$ of a maximal $\PSLG$ $\candidate_i$, the algorithm generates two particular types of chains. We refer to these as \emph{type $1$-chains} and \emph{type $2$-chains}. Let $\vertseq(f)=\langle v_0, v_2, \ldots, v_{m-1}\rangle$ be the boundary vertex sequence of $f$ that is being processed.
  A chain from $v_i$ to $v_k$ is a type \emph{$1$-chain}, provided (i) it has exactly one \convex vertex $v_j$ in its interior, i.e.,  $i < j < k$, (ii) every vertex participating in the chain is contained in $\N(v_j)$, (iii) every vertex from $v_{j+1}$ to $v_k$ is visible to $v_{j-1}$ and symmetrically, every vertex from $v_i$ to $v_{j-1}$ is visible to $v_{j+1}$ in $\vertseq(f)$.
  A chain from $v_i$ to $v_k$ is a type \emph{$2$-chain} provided (i) it has exactly two consecutive convex vertices, $v_j$ and $v_{j+1}$ such that $i+1 < j+1 < k$, and (ii) the chain from $v_i$ to $v_{j+1}$ is a type $1$-chain and the chain from $v_j$ to $v_k$ is a type $1$-chain (see Figure~\ref{fig:one-two-chains}). For brevity, we refer to these as $1$- and $2$-chains.

Let $\chain(v_i, v_j)$ be a $1$-chain or a $2$-chain. Define $E(v_i, v_j)$ to be the interior of the region bounded by $\chain(v_i,v_j)$ and the segment $\overline{v_iv_j}$ in $\vertseq(f)$ (Figure~\ref{fig:one-two-chains}). By this definition, the chain $\chain(v_i,v_j)$ and the edge $\overline{v_iv_j}$ form the boundary of $E(v_i,v_j)$, however, this boundary is not included in the region $E(v_i,v_j)$. The effectiveness of the algorithm partially rests on the fact that $E(v_i,v_j)$ can be triangulated in a straight-forward fashion, provided it does not have any other input points in its interior. In Figure~\ref{subfig:one-chain}, $\chain(v_1, v_6)$ is a $1$-chain with $v_5$ the only convex vertex in the chain, and the region $E = E(v_1, v_6)$ can easily be triangulated. In Figure~\ref{subfig:two-chain}, $\chain(v_1, v_7)$ is a $2$-chain with $v_3$ and $v_4$ the two consecutive convex vertices, and the region $E = E(v_1, v_7)$ can also be triangulated as shown.

In the following section, we present our algorithm to compute an approximate Minimum Weight triangulation.


\section{Algorithm}
\label{sec:algorithm}

The algorithm iteratively constructs a candidate maximal $\PSLG$ $\candidate_i$ which is a subgraph of $\adjgraph{i}$ for each level  $i$. Each $\candidate_i$ can be seen as a partial triangulation of $P$. Initialize $\candidate_0 \leftarrow \emptyset$. We construct $\candidate_{i+1}$ from $\candidate_i$ in two phases. In the first phase, we process every face $f$ of $\candidate_i$ in the clockwise order, compute certain $1$-chains and $2$-chains, and triangulate their regions by adding some of the edges of $\leveledges{i+1}$ which results in an intermediate $\PSLG$ $\hat{\candidate_i}$. A chain generated by the algorithm starting at the vertex $v_k$ is denoted by $\chain(v_k)$, omitting the last vertex in the chain.  In the second phase, we process all the edges of $\leveledges{i+1}$  in a greedy fashion and construct a maximal $\PSLG$ $\candidate_{i+1} \subset \adjgraph{i+1}$. Next, we describe the details of the construction $\candidate_{i+1}$ from $\candidate_i$.

\paragraph{Phase 1:}
Initialize $\hat{\candidate}_i \leftarrow \candidate_i$. We describe our algorithm for each non-triangulated face $f$ of $\candidate_i$. 
If the boundary of $f$ contains more than one connected component, we repeat the algorithm for the  vertex sequence of each connected component of $f$. Let $\vertseq(f)=\langle v_0,\ldots,v_{m-1}\rangle$ be any such vertex sequence. 
We begin by choosing a start vertex, and then proceed to describe how each vertex is processed.

\textit{Choosing a start vertex.} Choose the starting vertex $v_j$ to be any backward support vertex in $\vertseq(f)$, i.e., $v_j = \back(v', \vertseq(f))$ for some $v' \in \vertseq(f)$. If there is no such vertex then let $v_j = v_0$. 

\textit{Processing $\vertseq(f)$.} Starting from $v_j$, we will \emph{visit} all the vertices as they appear in $\vertseq(f)$. If an edge incident on vertex $v$ is added, then $v$ is \emph{processed}.
Initially $k \leftarrow 0$ and $f' \leftarrow f$. As we visit the vertices, we add new edges, each of which splits $f$ into a triangle and a smaller non-triangulated face $f'$. We will dynamically maintain the vertex sequence  $\vertseq(f')$.

We process $v_{j+k}$ by repeating (1) and (2) until $k > m-1$.
\begin{enumerate}  
  \item For $v_{j+k},$ let $v_{l-1}$ be its forward convex vertex and let $v_{l}$ be its forward support vertex in $\vertseq(f)$. Mark $v_{j+k}$ as visited. If $v_{l}$ is visible to $v_{j+k}$ with respect to the face $f$ and 
$v_{l-1} \in \N(v_{j+k})$ then set the chain $\chain(v_{j+k})\leftarrow\chain(v_{j+k},v_l)$ (Figure~\ref{subfig:forward-chain}). Else set $\chain(v_{j+k}) = \chain(v_{j+k}, v_{j+k+1})$ (the  chain with only two vertices).

If $ v_{l-1}$ is not visited and $|\chain(v_{j+k})| > 2$, then  
\begin{enumerate}
  \item If $v_{l+1}$ is reflex, let $v_{q+1}$ be its backward convex vertex and $v_{q}$ be its backward support vertex. If $v_{q+1} \in \N(v_{l+1})$ and $v_{q}$ is visible to $v_{l+1}$  with respect to the face $f$, then  
  \begin{itemize}
    \item  Let $\chain(v_{l+1}, v_r)$ be the maximal reflex chain that starts at $v_{l+1}$.  Scan along this chain and identify the last vertex $v_s$ for which $v_{q+1} \in \N(v_s)$ and $v_{q}$ is visible to $v_s$ in $f$. Set $\chain(v_{j+k}) \leftarrow \chain(v_{j+k},v_s)$ (Figure~\ref{subfig:step1a}).
  \end{itemize}
  \item Now add an edge from $v_l$ to every vertex in $\chain(v_{j+k},v_{l-2})$.  If $\chain(v_{j+k})=\chain(v_{j+k},v_s)$, then we add edges from $v_{j+k}$ to every vertex in $\chain(v_{l+1},v_s)$. Let $\edges(v_{j+k})$ be the set of edges added. Set $k \leftarrow |\edges(v_{j+k})| + 1$. Mark every vertex in the chain $\chain(v_{j+k})$ as visited and update $f'$ by removing all vertices in the interior of $\chain(v_{j+k})$ from $\vertseq(f')$ (see Figure~\ref{fig:phase1}).
\end{enumerate}
Otherwise if $v_{l-1}$ is visited and $|\chain(v_{j+k})| > 2$, then
\begin{enumerate}
\item[(c)] Let $v_{j'}$ be the vertex appearing immediately after $v_j$ in $\vertseq(f')$. Set the chain $\chain(v_{j+k}) \leftarrow \chain(v_{j+k}, v_{j'})$ (Figure~\ref{subfig:start-vertex}). Add an edge from $v_{j'}$ to every vertex in $\chain(v_{j+k}, v_{l-2})$. Mark every vertex in the chain $\chain(v_{j+k})$ as visited and update $f'$ by removing all the vertices in the interior of $\chain(v_{j+k})$ from $\vertseq(f')$. Set $k \leftarrow m$.
\end{enumerate}
  
  \item If no edges are added in step (1), then set $k \leftarrow k+1$. 
\end{enumerate}

\begin{figure}
  \centering
  \begin{subfigure}{0.3\textwidth}
    \centering
    \includegraphics[width=0.8\textwidth]{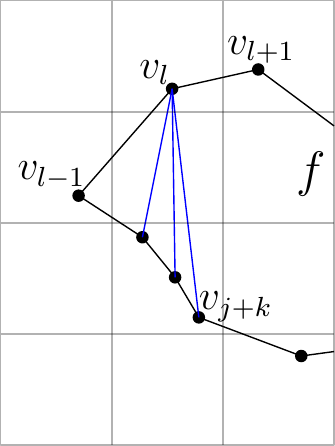}
    \caption{}
    \label{subfig:forward-chain}
  \end{subfigure}
  \begin{subfigure}{0.25\textwidth}
    \centering
    \includegraphics[width=0.85\textwidth]{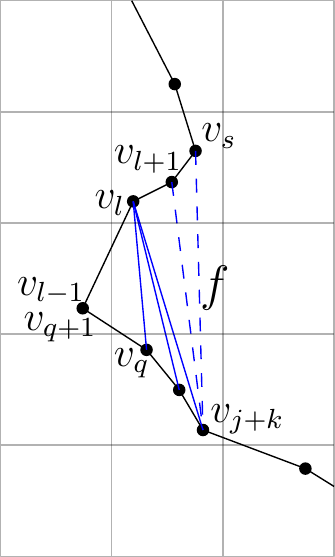}
    \caption{}
    \label{subfig:step1a}
  \end{subfigure}
  \begin{subfigure}{0.3\textwidth}
    \centering
    \includegraphics[width=0.8\textwidth]{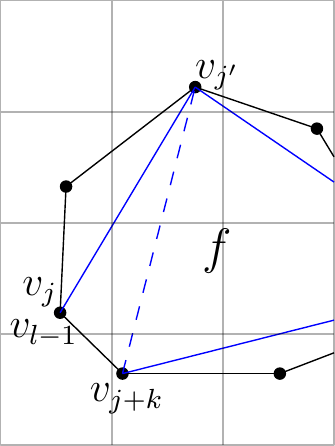}
    \caption{}
    \label{subfig:start-vertex}
  \end{subfigure}
  \caption[Edges added in Phase $1$ of the algorithm.]{Edges added in Phase $1$ (in blue). (a) $\chain(v_{j+k}) = \chain(v_{j+k}, v_l)$ and $v_{l+1}$ is not reflex. (b) $v_{l+1}$ is reflex, so step 1(a) is executed and $\chain(v_{j+k}) = \chain(v_{j+k}, v_s)$. (c) $v_{l-1}$ has already been visited and $v_{l-1} = v_j$, so $\chain(v_{j+k}) = \chain(v_{j+k}, v_{j'})$.}\label{fig:phase1}
\end{figure}

\paragraph{Phase 2:}
Initialize $\candidate_{i+1} \leftarrow \hat{\candidate}_i$. For each edge $e$ in $\leveledges{i+1}$, process them in increasing order of their Euclidean lengths. If $e$  has not yet been added to $\candidate_{i+1}$ and if $\candidate_{i+1} \cup e$ is still a $\PSLG$, add $e$ to $\candidate_{i+1}$. Once all edges of $\leveledges{i+1}$ are processed, the resulting graph $\candidate_{i+1}$  is a maximal $\PSLG$ of $\adjgraph{i+1}$.

These two phases are repeated for each level $i$, until $\candidate_{\log_3{9\Delta}+1}$, a maximal $\PSLG$ on the complete graph of the input set (and hence a triangulation), is computed, which yields our approximate solution. Note that the number of non-empty levels  for the edges can not exceed $O(n^2)$ and for each level, the processing time is polynomial. Therefore, the execution time of our algorithm is a polynomial independent of $\Delta$.  

  While processing a vertex $v_t$, suppose the algorithm executes step 1(c). Note that this case only arises when the forward convex vertex of $v_t$ ($v_{l-1}$ for that step) is either the start vertex $v_j$ or the following vertex $v_{j+1}$. At the end of step 1(c), we set $k = m$ and the algorithm completes executing Phase 1 for the face $f$. Alternatively, if the forward support vertex $v_l$ of $v_t$ is the start vertex $v_j$, then Phase 1 for the face $f$ will terminate without executing 1(c). Call $v_t$, the last vertex to be processed in this sequence, the end vertex. The following observations follow from the description of the algorithm: 

\begin{description}
  \item[(O1)] For any pair of vertices $v,v'$ processed by the algorithm, the chains  $\chain(v)$ and $\chain(v')$ are interior disjoint except for the pair corresponding to the start vertex $\chain(v_j)$ and the end vertex $\chain(v_t)$. For $v_j$ and $v_t$, either $\chain(v_j)$ is contained inside $\chain(v_t)$ (if step 1(c) was executed) or $\chain(v_j)$ and $\chain(v_t)$ are interior disjoint ($v_j$ was the forward support of $v_t$).
  \item[(O2)] The last vertex $v_{k'}$ on any maximal reflex chain $\chain(v_{k},v_{k'})$ will have an edge to its forward support $v_l$ provided the forward convex vertex $v_{l-1} \in \N(v_{k'})$ and the vertex $v_l$ is visible to $v_{k'}$ with respect to the face $f$
\end{description}

\section{Approximation Ratio}
\label{sec:approx-ratio}

In this section, we begin by showing that the algorithm presented in Section~\ref{sec:algorithm} produces, in the worst case, a $24$-approximation of the MWT. We will also show that the expected approximation ratio of our algorithm is $16$.
Our algorithm maintains two invariants. To state the invariants, for every level $i$, we first introduce the \emph{disc graph} $\ballgraph{i}$, defined as follows. Let $P$ be the vertex set of $\ballgraph{i}$, and given two vertices $u, v \in P$, $uv \in \ballgraph{i}$ if and only if $\|uv\| \le \frac{\rand3^{i-1}}{\sqrt{2}}$. In other words, the point $v$ is contained in a disc that is centered at $u$ with a radius of $\frac{\rand3^{i-1}}{\sqrt{2}}$. Clearly $\ballgraph{i} \subseteq \adjgraph{i}$ for all $i$.
Furthermore, let $\opt$ be the optimal triangulation of $P$, and let $\opt_i$ be the edges of the optimal triangulation that are in the disc graph $\ballgraph{i}$, i.e., $\opt_i = \opt \cap \ballgraph{i}$. 

With these definitions, we can state the invariants  maintained by the algorithm  for every level $i \in \{0, \ldots, \log_3{9\Delta}+1\}$.

\begin{invariant}
\label{inv:criticallength}
$\hat{\candidate}_i$ is a $\PSLG$ and every edge $uv \in \hat{\candidate}_i \setminus \candidate_i$ has length $\|uv\| \leq 4\sqrt{2} \cdot\rand3^{i-1}$.
\end{invariant}

\begin{invariant}
\label{inv:cardinality}
$|\hat{\candidate}_i| \geq |\opt_{i}|$.
\end{invariant}

Note that the length bound in Invariant~\ref{inv:criticallength} for edges added in Phase 1 of the algorithm implies that all of the edges added are in $\leveledges{i+1}$. Assuming both invariants hold, we will bound the approximation ratio of our algorithm.

\begin{theorem}[Approximation Ratio]
\label{thm:approx}
Suppose Invariants~\ref{inv:criticallength} and~\ref{inv:cardinality}
hold. Then the candidate solution produced by the algorithm in
Section~\ref{sec:algorithm}, $\candidate_{{\log_3{9\Delta}+1}}$, is an
$\alpha$-approximate MWT, where $\alpha \le 24$. Furthermore, the
expected value of $\alpha$ is at most $16$.
\end{theorem}
\begin{proof}
Suppose that both the invariants hold. Let $\candidate =
\candidate_{{\log_3{9\Delta}+1}}$ be the triangulation computed by our
algorithm.  Any triangulation for a given point set $P$ has $m=3n-3-h$
edges, where $h$ is the number of points on the convex hull of $P$.
Therefore, $|\candidate| = |\opt| = m$. Let $\tau =\langle a_1,\ldots,
a_m\rangle$ be the ordering of edges of $\candidate$ based on when
they were added by the algorithm, i.e., edge $a_i$ appears before edge
$a_j$ in $\tau$ if $a_i$ was added before $a_j$ by the algorithm.
Note that in the sequence $\tau$, the edges of level $i$ appear before
all level $i+1$ edges. Let $\optdiff{i} = \opt_i \setminus
\opt_{i-1}$. We also order the edges of the optimal triangulation
$\opt$ into another sequence $\tau^*=\langle t_1,\ldots, t_m\rangle$.
In this sequence, for any $i<j$, the edges of $\optdiff{i}$ appear
before the edges of $\optdiff{j}$. Let $t_j$ be the $j^{\text{th}}$ edge in $\tau^*$, and $a_j$ be the $j^{\text{th}}$ edge of $\tau$.
To prove this theorem, we will show that $\alpha_j=\frac{\wt(a_j)}{\wt(t_j)} \le 24$ and $\mathbb{E}[\alpha_j]\le 16$.
Given this, we obtain the bound on $\alpha$ as follows:
\begin{eqnarray}
\alpha = \frac{\wt(\candidate)}{\wt(\opt)} &=&
\sum_{j=1}^m\frac{\wt(t_j)}{\wt(\opt)}\cdot \frac{\wt(a_j)}{\wt(t_j)}
= \sum_{j=1}^m\beta_j\cdot \frac{\wt(a_j)}{\wt(t_j)},
\label{eq:approx}
\end{eqnarray}
where $\beta_j = \frac{\wt(t_j)}{\wt(\opt)}$. Since $\beta_j > 0$ and
$\sum_{j=1}^m \beta_j =1$,  $\alpha$ is a weighted average of all
$\alpha_j$ values. Therefore, we can bound $\alpha$ and
$\mathbb{E}[\alpha]$ by providing an upper bound for $\alpha_j$ and
$\mathbb{E}[\alpha_j]$.

  From here on, we will bound $\alpha_j$ by $24$ and
$\mathbb{E}[\alpha_j]$ by $16$ to prove the theorem. Let $k$ be an
integer such that $t_j \in \optdiff{k}$. Therefore, the cost of $t_j$,
$w(t_j) \ge \frac{\rand3^{k-2}}{\sqrt{2}}$.
By Invariant~\ref{inv:cardinality}, it follows that $j \leq
|\hat{\candidate}_{k}|$, so $a_j \in \hat{\candidate}_k$. By
Invariant~\ref{inv:criticallength}, $\wt(a_j) \leq \rand4\sqrt{2}
\cdot 3^{k-1}$.
Therefore,
\[ \dfrac{\wt(a_j)}{\wt(t_j)} \leq
   \dfrac{\rand 4\sqrt{2} \cdot 3^{k-1}}{\frac{\rand3^{k-2}}{\sqrt{2}}} = 24 \]
Thus, $\candidate{}$ is an $24$-approximation of $\opt$.

Recollect that $\rand$ is chosen uniformly at random from the
interval $(\frac{1}{3},1)$.
$\alpha_j$ may be expressed as a function of $\rand$ as follows. Let $i$ be the largest integer such that $\wt(t_j) \ge \frac{\rand 3^{i-1}}{\sqrt{2}}$. Let $t_j$ be the edge $pq$ in $\opt$. By the choice of $i$,
$\frac{\rand}{3}\le \frac{\sqrt{2}\cdot\|pq\|}{3^{i}} \le \rand$.
Since $\rand \in (\frac{1}{3},1)$, either $t_j \in \optdiff{i+1}$ or $t_j \in \optdiff{i+2}$.
Hence there are two possible cases:

\begin{itemize}
\item If $\rand \in(\frac{1}{3},\frac{\sqrt{2}\|pq\|}{3^{i}})$, then
$t_j \in \optdiff{i+2}$. From Invariant 2, $\wt(a_j$)  will be at
most $4\sqrt{2}(\rand 3^{i+1})$. Therefore $\alpha_j$ is
$\frac{4\sqrt{2}\rand 3^{i+1}}{\|pq\|}$.

\item If $\rand \in [\frac{\sqrt{2}\|pq\|}{3^{j}},1)$ then $t_j \in
\optdiff{i+1}$. From Invariant 2, $\wt(a_j$)  will be at most
$4\sqrt{2}(\rand 3^{i})$. Therefore $\alpha_j$ is
$\frac{4\sqrt{2}\rand 3^{i}}{\|pq\|}$
\end{itemize}

Let $x = \frac{\sqrt{2}\|pq\|}{3^{i}}$. The expected value of
$\alpha_j$ can be expressed as:
\begin{eqnarray*}
\mathbb{E}[\alpha_{j}] &=&  \frac{3}{2}\int_{1/3}^x
\left(\frac{4\sqrt{2}\rand 3^{i+1}}{\|pq\|}\right) \,\mathrm{d}\gamma
+ \frac{3}{2}\int_x^1 \left(\frac{4\sqrt{2}\rand 3^{i}}{\|pq\|}
\right)\,\mathrm{d} \rand\\
&=&  \frac{2\sqrt{2}3^{i+1}}{\|pq\|}\left(x^2+\frac{1}{3}\right) =
\frac{4\sqrt{2}\|pq\|}{3^{i-1}} + \frac{2\sqrt{2}3^i}{\|pq\|} \le 16
\end{eqnarray*}
The last inequality holds for any $p,q$ such that $\|pq\|\in
\left[\frac{3^{i-1}}{\sqrt{2}},\frac{3^i}{\sqrt{2}}\right]$, and hence for every $t_j$ being considered.
\end{proof}



\section{Algorithmic Invariants}
\label{sec:invariants}

To complete the proof of Theorem~\ref{thm:approx}, we prove Invariant~\ref{inv:criticallength} in Section~\ref{subsec:proof-inv-1} and Invariant~\ref{inv:cardinality} in Section~\ref{sec:inv2}, which were assumed in the proof of the approximation ratio. It will be useful to first consider properties of a maximal $\PSLG$ with respect to $\adjgraph{i}$ which arise due to the underlying grid structure. These properties can be used to prove Invariant~\ref{inv:criticallength}.

\subsection{Properties of a Maximal $\PSLG$ in $\adjgraph{i}$}
\label{subsec:grid-properties}
In order to state the properties satisfied by any maximal $\PSLG$ of $\adjgraph{i}$, we introduce the following definitions. For any two points $u$ and $v$, let $C_u$ and $C_v$ be the cells of $\grid{i}$ that contain $u$ and $v$. Let $\Gamma_i^0(u) = \{C_u\}$ and let $\Gamma_i^j(u) = \bigcup_{C \in \Gamma_i^{j-1}} N(C)$. We say that $u$ and $v$ are \emph{$k$ cells apart} if $C_v\in\Gamma_i^{k-1}$. It is not difficult to see that if $u$ and $v$ are $k$ cells apart,  then  $\|uv\|\le k\sqrt{2}\cdot \rand3^{i-1}$ and the worst case is achieved when $u$ and $v$ are diagonally opposite corners of the square of side length $k\rand3^{i-1}$ containing $k^2$ cells of $\grid{i}$.

Next, we state the properties. The proofs of each of these properties are deferred until Section~\ref{subsec:proof-grid-props}.
  Let $\candidate$ be a maximal $\PSLG$ with respect to $\adjgraph{i}$. Let $f$ be a non-triangulated face of $\candidate$ with a boundary vertex sequence $\vertseq(f) =  \langle v_0, \ldots, v_{m-1} \rangle$.

\begin{description}
  \item[(P1)]  Suppose $v_j \in \vertseq(f)$ is convex. Then $v_{j-1}$, $v_j$, and $v_{j+1}$ are in three distinct cells $C_{v_{j-1}}$, $C_{v_j}$, and $C_{v_{j+1}}$ of $\grid{i}$, respectively, and $C_{v_{j-1}} \notin N(C_{v_{j+1}})$.
  \item[(P2)] Suppose $v_{j} \in \vertseq(f)$ and $v_{k}v_{k+1}$ is any edge on the boundary of $f$ such that $v_k \in \N(v_{j})$ (resp. $v_{k+1} \in \N(v_j)$) and $v_{k}v_{k+1}$  is visible to $v_{j}$ in $\vertseq(f)$. Then,
  \begin{itemize}
    \item[(i)] the chain $\chain$ from $v_j$ to $v_{k+1}$ (resp. chain $\tilde{\chain}$ from $v_k$ to $v_j$) in $\vertseq(f)$ is a $1$-chain, and,
    \item[(ii)] $v_{k+1}$ is a forward (resp. $v_k$ is a backward) support vertex for  every vertex from $v_j$ to $v_{k-1}$ (resp. from $v_{k+2} $ to $v_j$).
  \end{itemize}
  \item[(P3)] For any chain $\chain(v, y)$ from $v$ to $y$ in $\vertseq(f)$,   
  \begin{itemize}
    \item[(i)] if $\chain(v, y)$ is a $1$-chain with $v'$ as the only convex vertex in its interior, then the region $E = E(v, y)$ is contained in $\N(v')$, i.e., $E \subset \N(v')$, and $E$ contains no input points of $P$. 
    \item[(ii)] if the chain $\chain(u, y)$ is a $2$-chain with $v$ and $x$ as the two convex vertices, then the region $E = E(u, y)$ is such that $(E\cap(\N(v)\cup \N(x)))\cap P $ contains no points of $P$. In two cases, $E\not\subset(\N(v)\cup \N(x))$ (See Figure~\ref{fig:p3-four-cell-exception}) and may contain  points of $P$. In all other cases, $E$ contains  no  points of $P$.
  \end{itemize}
  \item[(P4)] For any vertex $v \in \vertseq(f)$, if an edge $xy$ on the boundary of $f$ is $\delta$-visible for $\delta=\frac{\rand3^{i-1}}{\sqrt{2}}$, then exactly one of $x$ and $y$ are in $\N(v)$ and the other is not.
\end{description}
 
\subsection{Proof of Invariant 1}
\label{subsec:proof-inv-1}

Using (P1)--(P4), we now establish Invariant~\ref{inv:criticallength}. To assist with the presentation, we re-use the notation from the description of the algorithm. Recollect that $v_j$ is the start vertex, $v_t$ is the end vertex  and while processing a vertex $v_{j+k}$, $v_{l-1}$ and $v_l$ are its forward convex vertex and forward support vertex, respectively. If $v_{l+1}$ is reflex, then $v_{q+1}$ and $v_{q}$ are the backward convex vertex and the backward support vertex, respectively. For every edge $uv$ added in Phase 1, let $\chain(u,v)$ in $\vertseq(f)$ be the chain from $u$ to $v$. We refer to this kind of chain $\chain(u,v)$ as a \emph{triangulated chain}. If $uv$ appears on the boundary of a non-triangulated face of $\hat{\candidate}_i$ after Phase 1 has been executed, then $\chain(u, v)$ is a \emph{maximal} triangulated chain. If the chain $\chain(u, v)$ was generated when processing $u$, we may omit the the other endpoint of the chain, and simply denote the chain by $\chain(u)$. 

Invariant~\ref{inv:criticallength} asserts two claims, (i) any edge $uv \in \hat{\candidate}_i \setminus \candidate_i$ has length $\|uv\| \leq 4 \sqrt{2}\cdot \rand 3^{i-1}$, and (ii) that the intermediate graph $\hat{\candidate}_i$ produced by Phase $1$ is a $\PSLG$. Note that to prove the first claim (i), it is sufficient to show that the endpoint of any edge $uv$ that was added by the algorithm is at most $4$ cells apart. We first show that any two vertices in a $1$-chain are at most $3$ cells apart and any two vertices in a $2$-chain are at most four cells apart (Lemma~\ref{lem:2-or-3-cells-apart}). We then show that any triangulated chain generated by the algorithm is either a $1$-chain or a $2$-chain (Lemma~\ref{lem:triangulated-chains-1-2-chains}). This proves (i). We show (ii) that planarity is never violated by edges added in Phase 1 in Lemma~\ref{lem:triangulated-chains-empty} and~\ref{lem:planarity}. It follows that the algorithm triangulates the region $E$ of each triangulated chain and therefore that the algorithm maintains Invariant~\ref{inv:criticallength}.
\begin{lemma}
\label{lem:2-or-3-cells-apart}
If $u,v$ are vertices in a $1$-chain, then $u$ and $v$ are at most $3$ cells apart. If $u,v$ are vertices in a $2$-chain, then $u$ and $v$ are at most $4$ cells apart.
\end{lemma}
\begin{proof}
Suppose $u$ and $v$ are two vertices of a $1$-chain $\mathcal{C}$. Let $C_u$ and $C_v$ be the cells of $\grid{i}$ containing $u$ and $v$ respectively. Let $v'$ be the only convex vertex in the interior of $\mathcal{C}$. By definition of $1$-chain, $u,v \in \N(v')$, therefore $u \in \N(v')$ and $v' \in \N(v)$. So $C_u \in \Gamma^2_i(v)$, therefore, by definition $u$ and $v$ are at most $3$ cells apart.
 
Suppose $u$ and $v$ are vertices in a $2$-chain with $v'$ and $v''$ as the two convex vertices. It follows from the definition of a $2$-chain that $u, v \in \N(v')\cup \N(v'')$ and $v' \in \N(v'')$. Therefore, any two cells of $\N(v')\cup \N(v'')$ are at most $4$-cells apart and therefore $u$ and $v$ are at most $4$ cells apart.
\end{proof}

  \noindent 
The following lemma establishes a useful property of $1$- and $2$-chains.
\begin{lemma}
\label{lem:chain-intersection}
For any $1$- or $2$-chain $\chain(x,y)$, consider a line segment $\overline{uv}$ that enters and exits the region $E(x,y)$.  Then, $\overline{uv}$ also intersects  a vertex or an edge in the chain $\chain(x,y)$.   
\end{lemma}
\begin{proof}
Note that the region $E(x,y)$ represents the interior of the connected region that is bounded by the chain $\chain(x,y) $ on one side and the edge from $x$ to $y$ on the other.
The segment $\overline{uv}$ cannot enter and exit the region $E(x,y)$ through the same boundary edge of $E(x,y)$. Therefore,  either the entry or exit point lies on the edge of $xy$ and the other point lies on a vertex or an edge of the chain $\chain(x,y)$. 
\end{proof}

During Phase 1, a vertex $v_{j+k} \in \vertseq(f)$ is processed and we add edges if and only if Step 1 of the algorithm is executed. Otherwise, if Step 2 is executed,  we do not add any edge and generate a trivial chain $\chain(v_{j+k})=\chain(v_{j+k},v_{j+k+1})$. Next we show that when Step 1 of the algorithm is executed, the  chain $\chain(v_{j+k})$\ generated by the algorithm is either a $1$- or a $2$-chain. 

\begin{lemma}
\label{lem:triangulated-chains-1-2-chains}
For every vertex $v_{j+k}$ processed by the algorithm, the  chain $\chain(v_{j+k})$ generated  is either a $1$-chain or a $2$-chain. Furthermore, the triangulated chain $\chain(v_j)$ for the start vertex is a $1$-chain.
\end{lemma}
\begin{proof}
Note that for any vertex to be processed, Step 1 of the algorithm must be executed. First, we show that the  chain $\chain(v_j)$ for the start vertex $v_j$ is a $1$-chain. The start vertex $v_j$ is chosen to be a backward support vertex, unless none exists. Suppose $v_j$ is a backward support vertex for some vertex $v'$, so $v_j = \back(v', \vertseq(f))$. By definition, the vertex that appears after $v_j$, $v_{j+1}$, is a convex vertex (the backward convex vertex of $v'$) and $v_{j+2}$ is a reflex vertex. However, it  also follows from the definition that $v_{j+1} = v_{l-1}$ is the forward convex vertex of $v_j$ and $v_{j+2} = v_l = \fwd(v_j, \vertseq(f))$ is the forward support vertex of $v_j$. From the precondition of step 1, $v_{j+1} \in \N(v_j)$ and $v_{j+2}$ is visible to $v_j$ therefore $\chain(v_j) = \chain(v_j, v_{j+2})$.
 The point $v_j$ visible to $v_{j+2}$ and therefore the edge $v_{j+1}v_{j+2}$ is also visible so by (P2), the chain $\chain(v_j, v_{j+2})$ is a $1$-chain. If step 1(a) is not executed, then we are done. Suppose step 1(a) is executed, so $\chain(v_j) = \chain(v_j, v_s)$, where $v_s$ is the last reflex vertex on the reflex chain in which $v_{j+2}$ appears, such that $v_j$ is visible to $v_s$ and $v_s \in \N(v_{j+1})$. Since all vertices from $v_{j+3}$ to $v_s$ are reflex and from (P2), it follows that $\chain(v_j)$ is a $1$-chain. Therefore, in either case $\chain(v_j)$, the chain generated for the start vertex, must be a $1$-chain.
  Now suppose there are no reflex vertices in $\vertseq(f)$, then there is no backward support vertex and we choose  $v_j = v_0$  . In this case, $v_{j+1}$ is the forward convex vertex and $v_{j+2}$ is the forward support that is also a convex vertex. It follows that $v_{j+2}$ is visible to $v_j$, $\chain(v_j)$ has only one convex vertex in its interior, and $v_j, v_{j+2} \in \N(v_{j+1})$. Thus $\chain(v_j)$ is a $1$-chain.   
 
  Next, we prove that the lemma holds for the end vertex $v_t$, the last vertex that is processed in $\vertseq(f)$. Either step 1(c) is executed when processing $v_t$ or it is not. Suppose step 1(c) is executed when processing $v_t$ and the chain $\chain(v_t)$ is generated; since this is a special case we prove this separately. If step 1(c) is being executed, the forward convex vertex $v_{l-1}$ has already been visited when  $v_t$ is  processed. From precondition of Step 1, $v_{l-1} \in \N(v_t)$ and $v_{l}$ is visible to $v_t$. In this case, the maximal triangulated chain is $\chain(v_t,v_{j'})$ (where $v_{j'}$ is the vertex adjacent to $v_j$ in $f'$). Note that $v_{j+1}$ is convex (due to the choice of starting vertex), therefore, the forward convex vertex (which we know from the precondition of step 1(c) is already visited) is either $v_{l-1} = v_j$ or $v_{l-1} = v_{j+1}$.
  If $v_{l-1} = v_j$, then $v_j$ and $v_{j+1}$ are both convex. It follows from (P2) that $\chain(v_t,v_{j+1})$ is a $1$-chain, and by the above argument $\chain(v_j)=\chain(v_{j},v_{j'})$ is also a $1$-chain. Therefore, by definition $\chain(v_t)=\chain(v_t,v_{j'})$ is a $2$-chain.
  If $v_{j+1} = v_{l-1}$, then $v_j$ is reflex. We have already shown that $\chain(v_j) = \chain(v_j, v_{j'})$ must be a $1$-chain, therefore $\chain(v_j)$ contains only one convex vertex in its interior, and thus the chain $\chain(v_t,v_{j'})$ also contains only one convex vertex, namely $v_{j+1}$, in its interior. Since $v_t \in \N(v_{j+1})$ it follows that $\chain(v_j, v_{j'})$ is a $1$-chain. The remaining case for $v_t$, when step 1(c) is not executed, uses the same argument as for any other $v_{j+k}$; the argument for this case is provided next.
 
Finally, we prove that for any other vertex $v_{j+k}$, which is neither the start vertex $v_j$ nor   any other vertex for which step 1(c) is executed, the chain $\chain(v_{j+k})$ is either a $1$-chain or $2$-chain.  Step 1(c) is not executed. Therefore, step 1(a)-(b) must be executed. From the precondition to execute step 1, $v_{l-1} \in \N(v_{j+k})$ and $v_l$ is visible to $v_{j+k}$.   This implies that the edge $v_{l-1}v_l$ is visible to $v_{j+k}$ and from (P2) we can conclude that the chain $\chain(v_{j+k},v_l)$ generated by the algorithm with $v_{l-1}$ as the only convex vertex is a $1$-chain.  Suppose the precondition for step 1(a) is met, then $v_{q+1} \in \N(v_{l+1})$ and $v_q$ is visible to $v_{l+1}$. Therefore,  we can use the same argument to show that $\chain(v_q, v_{l+1})$ is a $1$-chain. Since all vertices from $v_{l+1}$ to $v_s$ are reflex, such that $v_{q+1} \in \N(v_{l+1})$ and $v_q$ is visible, by extension it follows that $\chain(v_{q}, v_s)$ is also a $1$-chain. Now there are only two possibilities for the backward convex vertex $v_{q+1}$ of the vertex $v_{l+1}$: either (i) $v_{q+1} = v_{l-1}$ or (ii) $v_{q+1} = v_{l}$. In case (i), the chain $\chain(v_{j+k})=\chain(v_{j+k},v_s)$ has only one convex vertex, $v_{q+1}=v_{l-1}$ and by construction every vertex of this chain is inside $\N(v_{q+1})$, so that $\chain(v_{j+k})$ is a $1$-chain. Otherwise, in case (ii), $v_l$ and $v_{l-1}$ are adjacent convex vertices; it then follows that $\chain(v_{j+k}) = \chain(v_{j+k},v_s)$ is a $2$-chain. This completes the proof.
\end{proof}

Lemma~\ref{lem:triangulated-chains-1-2-chains} implies that the algorithm only generates $1$-chains and $2$-chains in Phase 1, therefore by applying Lemma~\ref{lem:2-or-3-cells-apart} it follows that any edge added by the algorithm in Phase 1 has length at most $4 \sqrt{2}\cdot \rand 3^{i-1}$. Note that a useful observation in light of this proof is that the algorithm only generates $2$-chains if it executes step 1(a) or 1(c) of the algorithm. Furthermore, for a triangulated chain $\chain(u, v)$, since $\chain(u, v)$ is either a $1$-chain or a $2$-chain, it follows that the edges added for this chain $\edges(u, v)$ will not intersect with each other. In order to establish that $\hat{\candidate}_i$ is planar, we show that edges added for two distinct chains do not intersect, nor do they intersect with any edges in $\candidate_i$.

\begin{figure}
  \centering
  \begin{subfigure}{0.35\textwidth}
    \centering
    \includegraphics[width=\textwidth]{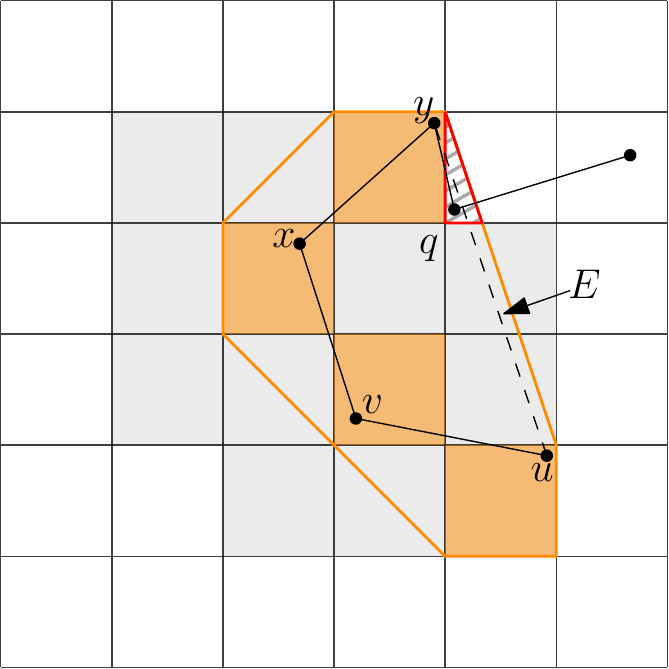}
    \caption{}
  \end{subfigure}\hspace{8mm}
  \begin{subfigure}{0.35\textwidth}
    \centering
    \includegraphics[width=\textwidth]{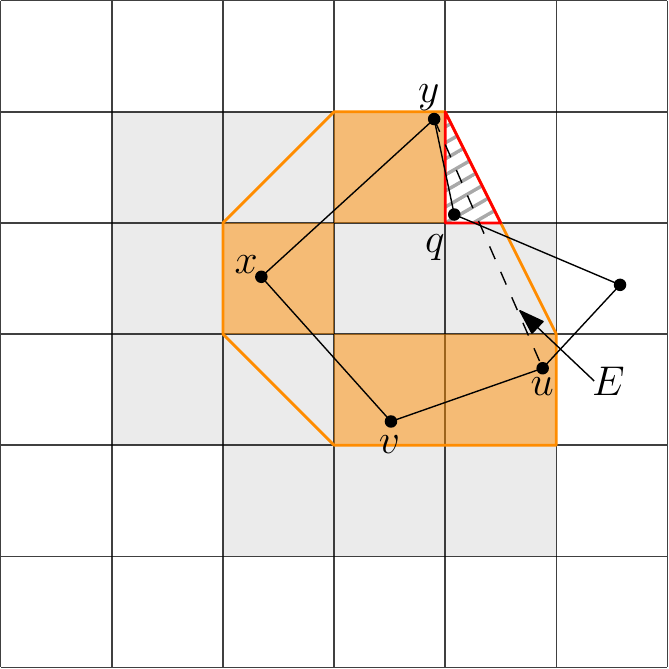}
    \caption{}
  \end{subfigure}
  
  \begin{subfigure}{0.35\textwidth}
    \centering
    \includegraphics[width=\textwidth]{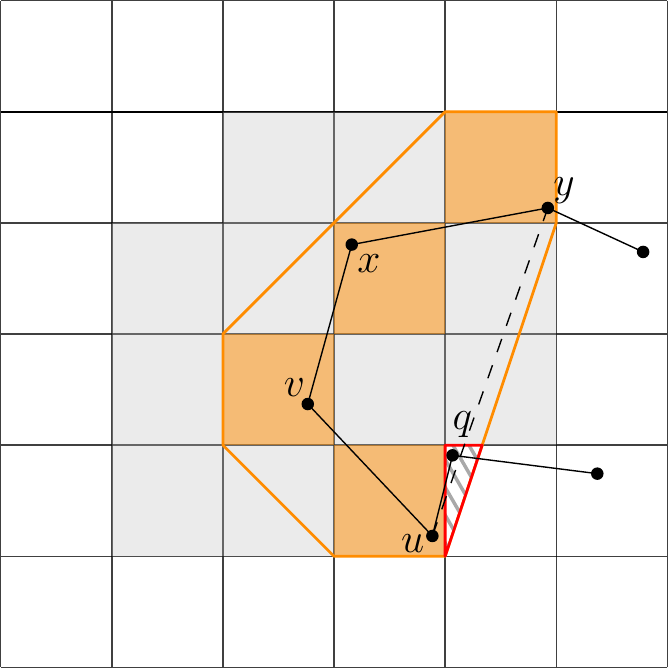}
    \caption{}
  \end{subfigure}\hspace{8mm}
  \begin{subfigure}{0.35\textwidth}
    \centering
    \includegraphics[width=\textwidth]{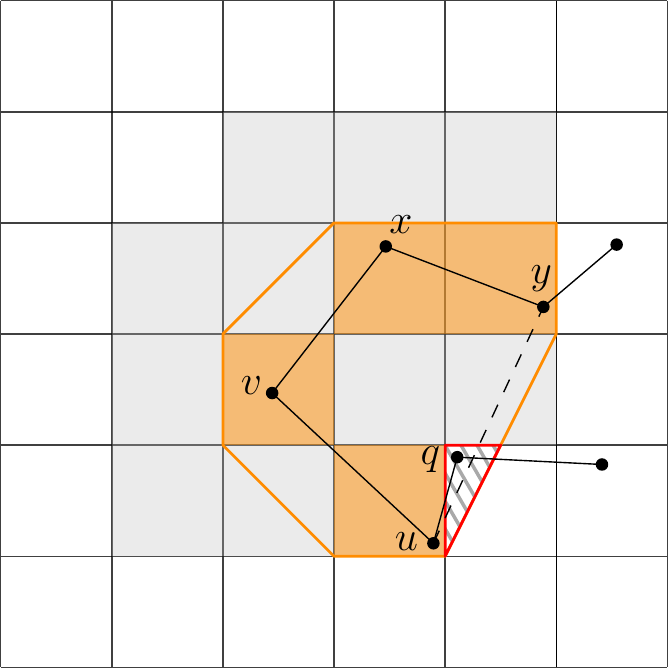}
    \caption{}
  \end{subfigure}
  \caption[Exceptional four-cell configurations for (P3).]{Exceptions in (P3): two possible configurations depicted by (a) and (b) of four cells such that $E(u, y)$ is not contained in $\N(v) \cup \N(x)$ ($N(v) \cup N(x)$ is depicted in grey). The vertex $q \in E(u, y)$. (c) and (d) are mirror images of (a) and (b), respectively.}\label{fig:p3-four-cell-exception}
\end{figure}
\begin{lemma}
\label{lem:triangulated-chains-empty}
Let $\chain(u,y)$ be a triangulated $2$-chain generated by the algorithm while processing the vertex $u$ and let $v$ and $x$ as the two convex vertices in its interior. Then, the region $E(u, y)$ bounded by $\chain(u,y)$ and the line segment $\overline{uy}$ contains no  points of $P$.
\end{lemma}
\begin{proof}
For the sake of contradiction, let us assume that there is a vertex $u \in \vertseq(f)$ such that when $u$ is processed, the algorithm generates a maximal triangulated $2$-chain $\chain(u,y)$ and the region $E(u, y)$ contains at least one input point $q \in \vertseq(f)$. The chain $\chain(u,y)$ appears on the boundary of $E(u,y)$ and so $q$ cannot be a vertex of $\chain(u,y)$.  By Lemma~\ref{lem:triangulated-chains-1-2-chains},  $\chain(v_j)$ is a $1$-chain, therefore $u$ cannot be the start vertex $v_j$. Since $\chain(u,y)$ is a $2$-chain, it follows by Property (P3) that $\N(v) \cup \N(x)$ does not contain  $q$ and $\overline{uy} \not\subset \N(v)\cup\N(x)$ only if $\chain(u, y)$ is in one of the two possible configurations. Figure~\ref{fig:p3-four-cell-exception} depicts these two possible configurations through four representative examples (clockwise and counter-clockwise cases for each of the two configurations). 

For cases (a) and (b), we choose $q$ to be the first vertex that appears after $y$ in the sequence $\vertseq(f)$ and that is contained in the region  $q\in E(u,y)$. We will first show that $q$ is a reflex vertex. Suppose $q$ is convex, from (P1)  the neighbors of $q$ in $\vertseq(f)$ are not in the same cell as $q$ and from (P3) they are also not in $E(u,y)\cap( \N(v) \cup \N(x))$. Since the cell of $q$ together with $\N(v) \cup \N(x)$ together cover $E(u,y)$, the neighbors of $q$ will lie outside $E(u,y)$ and both these edges to the neighbors intersect the edge $uy$ which would make $q$ a reflex vertex. Therefore, we can assume $q $ to be a reflex vertex.   Since $q$ is a reflex vertex and by our choice of $q$ to be the first vertex that appears after $y$ in the sequence $\vertseq(f)$, the edge $yx$ is visible to $q$ and by construction $y \in \N(q)$. From (P2),  the chain $\chain(q,x)$ is a $1$-chain with $y$ as the only convex vertex in its interior.    Only step 1(a) and step 1(c) of the algorithm generate a $2$-chain. The last vertex of any  $2$-chain generated by step 1(a) or step 1(c) of the algorithm will be a reflex vertex. Since $y$ is a convex vertex, while processing $u$ the algorithm will never generate a $2$-chain $\chain(u,y)$ of  the two configurations depicted by cases (a) and (b). \ 

For cases (c) and (d), let $q$ be the last vertex that appears before $u$ in the sequence $\vertseq(f)$. We can use an argument similar to cases (a) and (b) and show that  $q$ is a reflex vertex, and the chain $\chain(q,v)$ is a $1$-chain with $u$ as the only convex vertex in its interior. Let $q'$ be the vertex that appears before $u$ in this chain. Clearly, $q'$ is the last reflex vertex in this chain.\   From  (O2),  the algorithm would have added an edge from $q'$ to $v$ and would not have processed $u$. Hence, the algorithm will never generate a $2$-chain $\chain(u,y)$ of the two configuration as depicted by cases (c) and (d). Therefore $\chain(u,y)$ cannot be either of the two configurations.

\end{proof}

\begin{lemma}
\label{lem:planarity}
For any $i$, $\hat{\candidate}_i$ is a $\PSLG$, after Phase $1$ of the algorithm has been completed.
\end{lemma}
\begin{proof}
We will show that while processing a non-triangulated face $f$,  the edges added do not intersect with each other or intersect with any edge of $\candidate_i$. Given any two maximal triangulated chains $\chain(v_k)$ and $\chain(v_{k'})$ from (O1), we know that these chains are interior disjoint. If $\chain(v_k)$ (resp. $\chain(v_{k'})$) is a $1$-chain, then by (P3), the region $E(v_k)$ (resp. $E(v_{k'})$) does not contain any  points of $P$. If $\chain(v_k)$ (resp. $\chain(v_{k'})$) is a $2$-chain, then from Lemma~\ref{lem:triangulated-chains-empty},  the region $E(v_k)$ (resp. $E(v_{k'})$) does not contain any points. Therefore, the regions $E(v_k)$ and $E(v_{k'})$ are   disjoint regions and no edges of $\candidate_i$  can intersect   these regions. Since the edges in $\edges(v_k) $  triangulate the region $E(v_k),$ they cannot intersect with the edges in $\edges(v_{k'})$ or with edges already in $\candidate_i$. It follows that $\hat{\candidate}_i$ is planar.    
\end{proof}

This completes the proof of Invariant~\ref{inv:criticallength}. In the remainder of this section we show that after Phase 1, the intermediate $\PSLG$ $\hat{\candidate}_i$ has a particular property that plays a critical role in proving Invariant~\ref{inv:cardinality}.

\subsection{$\delta$-Visibility in $\hat{\candidate}_i$}
\label{subsec:delta-visiblility-prop}

In the following Lemma we prove a critical property used in the proof of Invariant~\ref{inv:cardinality}, together with some consequences of this property.

\begin{lemma}
\label{lem:delta-visibility} 
Let $f$ be a non-triangulated face in $\candidate_i$ with vertex sequence $\vertseq(f)$. Let $\vertseq(f') = \langle v_{j_0}, \ldots, v_{j_p}\rangle$ be the resulting vertex sequence after Phase $1$ of the algorithm, and let $\delta = \frac{\rand 3^{i-1}}{\sqrt{2}}$.

Suppose $v \in \vertseq(f)$, and let $xy$ be an edge that is $\delta$-visible to $v$.  Then,
\begin{enumerate}
\item if $v$ is in the interior of the chain $\chain(v_{j_r}, v_{j_{r+1}})$, then it follows that either $\{x,y,v\} \subseteq \chain(v_{j_{r-1}},v_{j_r+1})$, or $ \{x,y,v\} \subseteq \chain(v_{j_r},v_{j_{r+2}})$,
\item if $v$ is on the boundary of the chain $\chain(v_{j_r},v_{j_{r+1}})$ and without loss of generality let $v=v_{j_r}$. Then it follows that $ \{x,y,v\}\subseteq \chain(v_{j_{r}},v_{j_{r+1}})$ or $\{x,y,v\} \subseteq \chain(v_{j_{r-1}}v_{j_r})$ 
\end{enumerate}
\end{lemma}
\begin{proof}

 Suppose $v \in \vertseq(f)$ and $xy$ is an edge on the boundary of $f$, so $x$ and $y$ are consecutive vertices in $\vertseq(f)$. This lemma claims that $v, x$ and $y$ either belong to the same maximal triangulated chain, or adjacent maximal triangulated chains (here adjacent means the chains share exactly one endpoint). While processing a vertex $v'$ the chain $\chain(v')$ generated by the algorithm is  contained inside a unique maximal triangulated chain. Therefore, it suffices if we show that $v,x$ and $y$ either lie in the same chain $\chain(v')$ or are in two chains $\chain(v')$ and $\chain(v'')$ where $v'$ and $v''$ are vertices that are processed consecutively by the algorithm.

Note that $x$ and $y$ must belong to the same triangulated chain. This follows from the fact that $x$ and $y$ are adjacent in $\vertseq(f)$ so there are only three possibilities: (i) $xy$ also appears on the boundary of $f'$, (ii) $x$ and $y$ are both in the interior of the same triangulated chain, or (iii) $x$ is an endpoint and $y$ is in the interior of the same chain (or \textit{vice verca}). It is clear that $x$ and $y$ still belong to the same chain in all three cases.
 Without loss of generality, assume one of the following two cases.  (a) as one walks from $v$ in $\vertseq(f)$, one encounters $x$ before $y$  or  (b) as one walks  along $\vertseq(f)$, we encounter $y$ followed by $x$ and then $v$.  In both these cases since $xy$ is $\delta$-visible from $v$, from (P4), $x \in \N(v)$. 

In  case (a),  the edge $xy$ is visible to $v$ and $x \in \N(v)$. From (P2), the chain $\chain(v,y)$ is a $1$-chain with $x$ as the only convex vertex in its interior. The algorithm may or may not execute step 1 for the vertex $v$. 

Now, suppose step 1  is executed for  $v$, since $\chain(v,y)$ is a $1$-chain, the precondition of step 1 for $v$ is met and $v$ is processed. By construction, the chain $\chain(v)$ will contain its forward convex vertex $x$ and its forward support vertex $y.$ Therefore $x$ and $y$ will belong to the same maximal triangulated chain as $v$ (This corresponds to case (2) in the lemma statement).   

Suppose Step 1 is not executed for  $v$ (corresponds to case (1) of the lemma statement) , then there is some vertex $v'$ that is processed by the algorithm where the chain $\chain( v')$ contains  $v$ in its interior. We can assume that $\chain(v')$ does not contain both $x$ and $y$, since otherwise, this lemma holds trivially. So,  $v$ is in the interior of the chain $\chain(v')$, and both $x$ and $y$ are not in the chain $\chain(v')$. By construction, the last vertex  of $\chain(v')$, $v''$,  will be an interior vertex of the $1$-chain $\chain(v,y)$. The chain  $\chain(v'',y)$ is either an edge (when $v''=x$) or  a $1$-chain (when $v''$ is a reflex interior vertex of $\chain(v,y)$). If $v''=x$, then $\chain(v'')$ will trivially contain $y$ and the lemma holds. If $v''$ is a reflex vertex of the chain $\chain(v,y)$, then $\chain(v'',y)$ is  also a $1$-chain and the precondition of Step 1 is met for $v''$. The vertex $v''$ is processed and   the chain $\chain(v'')$ will contain all vertices of the chain $\chain(v'',y)$ including  $x$ and $y$. As a result, $x,y$ and $v$ are contained in the  chains $\chain(v')$ and $\chain(v'')$ where $v'$ and $v''$ are processed consecutively by the algorithm as desired.  

In  case (b), since the edge $xy$ is visible to $v$ and $x \in \N(v)$, from (P2), the chain $\chain(y,v)$ is a $1$-chain with $x$ as the only convex vertex in its interior. Let $v'$ be a vertex processed by the algorithm such that the chain $\chain(v')$ contains $y$ and $x$. We assume that the vertex $v$ is not in the chain $\chain(v')$ since otherwise the lemma is trivially true.  

Note that $v'$ is processed and step 1 is executed for the vertex $v'$.  $y$ and $x$ are inside the chain $\chain(v')$ and  the last vertex of this chain lies strictly between $y$ and $v$. Using the notations from the algorithm's description, $v' = v_{j+k}$.  $x$ is a convex vertex and $x$ is included in the chain $\chain(v')$. Therefore, either $x$  has to be the vertex $v_{l-1}$ or the vertex $v_l$. Since every vertex between $x$ and $v$ are reflex vertices and $\chain(y,v)$ is a $1$-chain, the precondition of  step 1 and step 1(a) is satisfied and the chain $\chain(v')$ will also contain  all reflex vertices on the chain $\chain(y,v)$. By our assumption, $v$ is not included in $\chain(v')$. Therefore $v$ must be a convex vertex. 

Let $v''$ be the vertex that appears before $v$  and let $\hat{v}$ be the vertex after $v$ in $\vertseq(f)$. We will claim that $v''$ is the last vertex in the chain $\chain(v')$. This is because the chain  $\chain(y,v'')$  is  also a $1$-chain with $x$ as the only convex vertex in its interior. Since $v''$ is reflex, step 1(a) would have added $v''$ to the chain $\chain(v')$.  As $v$ is not included in $\chain(v')$, the vertex $v''$ will be the last vertex of $\chain(v')$.

 Next, the algorithm will  execute step 1 for the vertex $v''$. Since $v'' \in \N(v)$ and the segment $v\hat{v}$ is visible to $v''$ (as $v$ is a convex vertex), from (P2), the chain $\chain(v'',\hat{v})$ is a $1$-chain, i.e., $v''$ is visible to $\hat{v}$ and $v'' \in \N(v)$. The precondition of Step 1 is met for $v''$  and therefore $\chain(v'')$ will contain the chain $\chain(v'',\hat{v})$ implying that $v$ in its interior of $\chain(v')$ (This shows that when  $x,y$ and $v$ are in different chains,  $v$ will be in the interior of the adjacent chain (case (1) of lemma statement); otherwise $x,y$ and $v$ will be in the same chain (case (2) of lemma statement)). Therefore, $x,y$ and $v$ will be contained in chains $\chain(v')$ and $\chain(v'')$ where $v'$ and $v''$ are two consecutive vertices processed by the algorithm. Hence the maximal triangulated chain containing $v$ is either the same chain that contains $x$ and $y$, or is adjacent to the maximal triangulated chain containing $x$ and $y$, as claimed.
\end{proof}

The next lemma applies Lemma~\ref{lem:delta-visibility} to $\opt_i$, establishing a crucial property that helps us relate the cardinality of $\hat{\candidate}_i$ to the cardinality of $\opt_i$ in the following section. This lemma shows that for a non-triangulated face $f'$ in $\hat{\candidate}_i$, if an edge in $uv \in \opt_i$ intersects the face $f'$,  then as you walk along $\dir{uv}$, every time you will enter and leave $f'$ on  boundary edges that are adjacent to each other. As a corollary, it follows that distinct connected components of $\hat{\candidate}_i$ cannot be arbitrarily close to each other.

\begin{lemma}
\label{lem:intersection}
Fix a non-triangulated face $f'$ of $\hat{\candidate}_i$. For any edge $uv$ of $\opt_i$, let $\{\overline{a_1b_1},\ldots, \overline{a_kb_k}\}$ be the maximal pairwise-disjoint open line segments resulting from taking $\overline{uv} \cap f'$, where each $\overline{a_jb_j} \subset (\overline{uv} \cap f')$. Then, for each $\overline{a_jb_j}$, there are three consecutive vertices $v_{p},v_{p+1}$ and $v_{p+2}$ in $\vertseq(f')$ with $a_j \in \overline{v_{p}v_{p+1}}$ and $b_j \in \overline{v_{p+1}v_{p+2}}$, and $v_{p+1}$ is convex in $\vertseq(f')$. 
\end{lemma}
\begin{proof}
  Let $f$ be that non-triangulated face in $\candidate_i$ which after the execution of Phase 1 of the algorithm on $f$, produces the face $f'$.   Let $uv$ be any edge in $\opt_i$. Let $\delta = \frac{\rand 3^{i-1}}{\sqrt{2}}$. If $\overline{uv} \cap f' = \emptyset$, the statement holds trivially. So assume $\overline{uv} \cap f' \neq \emptyset$ and $\{\overline{a_1b_1}, \ldots, \overline{a_kb_k}\}$ is a set of line segments of $\overline{uv} \cap f'$ that we encounter as we walk from $u$ to $v$. Any line segment can enter and exit the face $f'$ through a vertex or an edge of the boundary of $f'$. First, we will show that every such entry point $a_j$ and exit point $b_j$ is not an input point of $P$ but an intersection point of $uv$ with an  edge of the boundary of $f'$. $u$ and $v$ are points in $P$ and no three points of $P$ are collinear. So, the points in $\{a_{1,}a_2,\ldots,a_{k}\}\cap \{P\setminus\{u,v\}\} =\emptyset$ and $\{b_1,\ldots,b_{k}\}\cap \{P\setminus\{u,v\}\} =\emptyset$. Next, to show that $a_j$ and $b_j$ are not points in $P$, it suffices if we show that  $a_1 \neq u$ and $b_k\neq v$. We claim that $a_1 \neq u$. For the sake of contradiction, suppose $a_1=u$, then since the segment $\overline{a_1b_1}$ intersects $f'$, $b_1$ has to lie on an edge $v_p v_{p+1}$ where neither $v_p$ nor $v_{p+1}$ is $u$. Therefore, the edge $uv$ enters the region $E(v_p,v_{p+1})$. Since the region $E(v_p,v_{p+1})$ does not contain any points of $P$ (as shown in the proof of Lemma~\ref{lem:planarity}), the vertex $v$  lies outside $E(v_p,v_{p+1})$. \ The edge $uv$ enters and exits the region $E(v_p,v_{p+1})$  and from Lemma~\ref{lem:chain-intersection}, the edge $uv$ also intersects the chain $\chain(v_p,v_{p+1})$. Since $uv \in \opt_i$, by definition, $\|uv\| \le \delta$ and therefore $u$ is $\delta$-visible  to some edge  of $f$ belonging to the chain $\chain(v_p,v_{p+1})$.  By Lemma~\ref{lem:delta-visibility}, it follows that $u$ must either be $v_p$ or $v_{p+1}$ leading to a contradiction. A symmetric argument can be used to show that $b_k \neq v$. Therefore, every entry point $a_j$ and every exit point $b_j$ is an intersection point of $uv$ with an edge of the boundary of $f'$.\   

Next, we will show that every entry point $a_j$ and exit point $b_j$ are points on adjacent segments of the boundary of $f'$.  Consider one such segment $\overline{a_jb_j}$. Then $\overline{a_jb_j}$ intersects two boundary edges of $f'$, say $v_pv_{p+1}$ and $v_qv_{q+1}$. For the sake of contradiction, assume that these segments are not adjacent and so the end points of these segments, $v_p$, $v_{p+1}$, $v_q$ and $v_{q+1}$ are distinct points. The regions $E(v_p,v_{p+1})$ and $E(v_q,v_{q+1})$ do not contain any points of $P $ (as shown in the proof of Lemma~\ref{lem:planarity}). Therefore, the edge $uv$ has to enter and exit the regions $E(v_p,v_{p+1})$ and $E(v_q,v_{q+1})$. From Lemma~\ref{lem:chain-intersection}, the segment $uv$ also intersects  the chains $\chain(v_p,v_{p+1})$ and $\chain(v_q,v_{q+1})$. Either $uv$ intersects with edges or vertices that appear in the two chains. We provide a proof for the case where $uv$ intersects with edges of both the chains; the same argument can also be applied if the intersection point is a vertex. Let  the edges $xy$ (resp. $x'y'$ ) of the chains $\chain(v_p,v_{p+1}) $ (resp. $\chain(v_q,v_{q+1}))$  be the segments that intersect with $uv$. By definition, any edge of $\opt_i$ has a length less than or equal to $\delta$. In particular, $\|uv\| \leq \delta$.
  Without loss of generality, assume $uv$ is oriented vertically, then one can imagine sliding $uv$ horizontally to the left or to the right. The vertical distance between $xy$ and $x'y'$ can increase, at most, in one direction (since we are considering straight-line edges). Assume the vertical distance increases or stays the same to the right, then slide $uv$ to the left (so that the vertical distance stays the same or reduces) until $uv$ intersects an input point $z$. Such a point exists, because if no other point is encountered, sliding $uv$ in this manner would encounter one of the endpoints of $xy$ or $x'y'$. The point $z$ has one of three possibilities: a)  $z$ is a point  in $\vertseq(f')$, or b) $z$ is a vertex in one of the two chains $\chain(v_p,v_{p+1})$  and $\chain(v_qv_{q+1})$ or c)  $z$ is in the interior of some chain $\chain(v_s,v_{s+1})$ that is not $\chain(v_p,v_{p+1})$ or $\chain(v_q,v_{q+1}) $.  In case (c), this sliding vertical line segment through $z$ has to enter and exit the region $E(v_s,v_{s+1})$. Since, $\chain(v_s,v_{s+1})$ is a $1$- or a $2$-chain, from Lemma~\ref{lem:chain-intersection},  this sliding line segment must also intersect then also intersect the chain $\chain(v_s,v_{s+1})$ which is a contradiction. Therefore, $z$ cannot be in the configuration given by case (c).  

For (a) , there is an edge in the chain $\chain(v_p,v_{p+1})$ and chain $\chain(v_qv_{q+1})$ of $f$ that is $\delta$-visible from $z$. As $z \in \vertseq(f')$, by Lemma~\ref{lem:delta-visibility}, it follows that $z$ would be part of both the  chains $\chain(v_p,v_{p+1})$ and $\chain(v_q,v_{q+1})$ in $\candidate_i$ implying that they are the same or adjacent chains. For case (b),  without loss of generality, let $z$ be a point in the chain $\chain(v_pv_{p+1})$. Since $x'y'$ is $\delta$-visible to $z$ and since $z \in \chain(v_pv_{p+1})$ and $x'y'$ is an edge in the chain $\chain(v_qv_{q+1})$, by Lemma~\ref{lem:delta-visibility},  the two chains share a common endpoint.  So, in $\hat{\candidate_i}$, $v_pv_{p+1}$ and $v_qv_{q+1}$ must either be the same edge, or adjacent on the boundary of $f'$ implying that they share at least one end point, say $v_{p+1}=v_q$ as claimed.  

 Next, we show that $v_{p+1}$ is a convex vertex. The segment   $\overline{a_jb_j}$ is inside $f'$ and therefore, it appears on the right as we walk from $v_p$ to $v_{p+1}$.  Since $b_j$ intersects $v_{p+1}v_{p+2}$ , we need to make a right turn at $v_{p+1}$ implying that the vertex $v_{p+1}$ must be a convex vertex.

\end{proof}

\begin{corollary}
\label{cor:components}
  Let $f'$ be a non-triangulated face in $\hat{\candidate}_i$, and $uv$ any edge in $\opt_i$. If $uv$ intersects any two boundary edges $e_1$ and $e_2$ of $f'$, then $e_1$ and $e_2$ belong to the same connected component of $\hat{\candidate}_i$.
\end{corollary}
\begin{proof}
  Suppose $f'$ is a non-triangulated face of $\hat{\candidate}_i$, $uv \in \opt_i$, but $uv$ intersects the boundary of $f'$. Let $\{\overline{a_1b_1}, \ldots, \overline{a_kb_k}\}$ be a set of line segments resulting from $f' \cap \overline{uv}$, where each $\overline{a_jb_j} \subset \overline{uv}$. Consider one such segment $\overline{a_jb_j}$. By Lemma~\ref{lem:intersection}, $a_j \in \overline{v_{p}v_{p+1}}$ and $b_j \in \overline{v_{p+1}v_{p+2}}$ for three consecutive vertices $v_{p}$, $v_{p+1}$, and $v_{p+2}$ in $\vertseq(f)$, and therefore $v_{p}v_{p+1}$ and $v_{p+1}v_{p+2}$ belong to the same component in $\hat{\candidate}_i$. Since this is the case for each $a_jb_j$, the result follows.
\end{proof}

Consider a subset $C \subseteq P$ such that vertices in $C$ belong to a single maximal connected component of $\hat{\candidate}_i$.
It follows from Corollary~\ref{cor:components}, that the input points in $C$ correspond to one or more maximal connected components in $\opt_i$. Thus if a vertex $p \in C$, then $p$ cannot belong to the same connected component of $\opt_i$ as any vertex $p' \notin C$. Note that vertices in $C$ may, however, form more than one maximal connected component in $\opt_i$. We can prove Invariant~\ref{inv:cardinality} by simply proving it for each connected component containing the vertex set $C$ in $\hat{\candidate}_i$ and the corresponding connected components  of $\opt_i$ containing the vertices of $C$. In  Section~\ref{sec:inv2}, we present our argument for one such connected component, and for simplicity, we will use $\hat{\candidate}_i$ to denote this component.


\section{Invariant 2}
\label{sec:inv2}

Invariant~\ref{inv:cardinality} compares the cardinality of the $\PSLG$ $\hat{\candidate}_i$ to the cardinality of the $\PSLG$ $\opt_i$. If the region triangulated by $\hat{\candidate}_{i}$ contains the region triangulated by $\opt_i$, then it is easy to show that $\hat{\candidate}_{i}$ has a greater number of edges than $\opt_i$.  However, as shown in Lemma~\ref{lem:intersection}, edges of $\opt_i$ can intersect non-triangulated faces of $\hat{\candidate}_i$ and therefore, these two $\PSLG$s triangulate different regions. Nevertheless, Lemma~\ref{lem:intersection} implies that edges of $\opt_i$ which intersect a non-triangulated face of $\hat{\candidate}_i$ $f$, will intersect two edges of $\hat{\candidate}_i$ that are adjacent on the boundary of $f$. In order to establish Invariant 2, we will show that the regions triangulated by $\opt_i$ and $\hat{\candidate}_i$ are ``close" (cf. Figure~\ref{fig:trace}) to each other.
In Section~\ref{subsec:cardinality} we provide conditions under which we can compare the cardinalities of any two $\PSLG$s.
After that, in Section~\ref{subsec:inv2-proof}, we show these conditions are satisfied for $\hat{\candidate}_i$ and $\opt_i$, which allows us to prove Invariant~\ref{inv:cardinality}.

\subsection{Comparing Cardinality}
\label{subsec:cardinality}

Let $\mathcal{G}$ be any planar graph. Let $F$ denote the set of  faces of this planar embedding of $\mathcal{G}$. When the graph being considered is not clear, we denote the set of faces by $F(\mathcal{G})$. For any face $f \in F$ and its vertex sequence $\vertseq(f)$, we define its \emph{signature}, $s(f)$, to be the length of the vertex sequence $\vertseq(f)$, i.e.,  $s(f) = |\vertseq(f)|$. Let $X$ be a connected planar graph and $Y$ be any planar graph. For any two faces $f_1 \in F(Y)$ and $f_2\in F(X)$, we say that $f_1$ \emph{dominates}  $f_2$ if and only if $s(f_1)\ge s(f_2)$. Suppose, for every non-triangulated face $f$ in $F(X)$, there is a unique dominating face in $Y$, then we will show that $|X|\ge |Y|$ (Corollary~\ref{cor:size}). In Section~\ref{subsec:inv2-proof}, we will use this to prove Invariant~\ref{inv:cardinality}.
\begin{lemma}
\label{lem:count}
Consider a connected planar graph $\mathcal{G}$ and let $F$ be all the faces in the planar embedding of $\mathcal{G}$. The total number of edges in the graph $|\mathcal{G}|$ can be written as
$$|\mathcal{G}|=3n-6 -\sum_{f \in F} (s(f)-3)$$
\end{lemma}
\begin{proof}
For any face $f \in F$, its  signature is the length of the vertex sequence. Recollect that we construct the vertex sequence by exploring the edges of the boundary of the face $f$  in the clockwise direction. Every edge $e$ has at most two faces, one on each side. We refer to these faces as the co-faces of $e$.    

By the construction of the vertex sequence, every edge contributes $1$ to the signature of each of its co-faces. If the edge $e$ has only $1$ co-face $f$, it contributes two to $s(f)$. Therefore, $\sum_{f\in F}s(f) = 2|\mathcal{G}|$. From Euler's formula, we know $|\mathcal{G}| = n+|F| -2$ and therefore, $3|\mathcal{G}|=3n-6+3|F|$. It follows that,
  \begin{align*}
    |\mathcal{G}|+\sum_{f\in F} s(f) &= 3n-6+3|F|, \text{ so }
    |\mathcal{G}|=3n- 6 - \sum_{f\in F} (s(f)-3).
  \end{align*}
\end{proof}

In Lemma~\ref{lem:count}, if the graph is disconnected, then we can extend the proof to show that the number of edges is strictly smaller than $3n-6 -\sum_{f \in F} (s(f)-3)$. In addition, note that for any triangular face $f$, $s(f)-3 =0$. Hence, the total number of edges in any connected planar graph can be calculated using only the size of the vertex sequences of the non-triangulated faces. Using this observation, we obtain the following:

\begin{corollary}
\label{cor:size}
Let $X$ be a connected planar graph and $Y$ be any planar graph. For every non-triangulated face $f \in F(X)$, suppose there is a unique non-triangulated face $f' \in F(Y)$ such that $s(f) \le s(f')$. Then, $|X| \ge |Y|$.  
\end{corollary}

\begin{proof}
Let $F^*(X)$ be the set of non-triangulated faces of $X$, and $F^*(Y)$ be the set of non-triangulated faces of $Y$. Applying Lemma~\ref{lem:count} and the observation that for every triangular face $f$ of $X$, $s(f)-3 =0$,  we can express the number of edges in $X$ as
$$|X| = 3n - 6 - \sum_{f\in F^*(X)}(s(f)-3).$$
Since, for every face $f\in F^*(X)$ there is a unique face $f'\in F^*(Y)$ such that $s(f')\ge s(f)$ and since $s(f')-3 \ge 0$, we have $\sum_{f\in F^*(X)}(s(f)-3) \le \sum_{f'\in F^*(Y)}(s(f')-3)$. Since $Y$ is not necessarily connected, using an almost identical argument to Lemma~\ref{lem:count}, we can express the number of edges in $Y$  by the following inequality,
  \begin{align*}
    |Y| &\le 3n-6- \sum_{f\in F^*(Y)}(s(f)-3) \\
        &\le 3n-6 - \sum_{f\in F^*(X)}(s(f)-3) \\
        &= |X|.
  \end{align*}
\end{proof}

It follows from Corollary~\ref{cor:size}, that in order to determine whether one can apply Lemma~\ref{lem:count}, one need only be concerned with the non-triangulated faces of a particular planar graph. Thus, from this point on, we let $F(X)$ denote only the non-triangulated faces of a planar graph $X$.

\subsection{Proving Invariant~\ref{inv:cardinality}}
\label{subsec:inv2-proof}

\paragraph{Our strategy.} To prove Invariant~\ref{inv:cardinality}, we will add edges to $\opt_i$ in such a way that $\opt_i$ dominates $\hat{\candidate}_i$ and then apply Corollary~\ref{cor:size}. Fix any non-triangulated face $f \in F(\hat{\candidate}_i)$. When we overlay $f$ on the straight line embedding of $\opt_i$, we will use  Lemma~\ref{lem:intersection} to show that $f$ has only one connected ``non-trivial" intersection with some non-triangulated face $f'$ of $\opt_i$. We define the region of interest $f\cap f'$ as the trace of $f$. Next, we augment the graph $\opt_i$ by embedding new edges (planar but not necessarily straight-line) and carefully create a new face $f''$ around the trace of  $f$. We show that $f''$ dominates $f$ and refer to $f''$ as the \emph{dominating face} of $f$. After this procedure is repeated for every non-triangulated face in $F(\hat{\candidate}_i)$, the augmented graph $\opt_i$ now dominates $\hat{\candidate}_i$. However, adding edges to $\opt_i$ may create multiple (duplicate) edges between the same pair of points and therefore, Corollary~\ref{cor:size} does not apply. However, we show that duplicate edges cannot participate in two distinct dominating faces. Removing one of the duplicate edges will merge two faces $h$ and $h'$ and create a new face $h''$ that has a signature greater than or equal to  the signature of $h$ or $h'$. Since either $h$ or $h'$ (and not both) can be a dominating face of some face $f \in F( \hat{\candidate}_i$),  $h''$ will be the unique dominating face of $f$.
After deleting the duplicate edges, we now have an augmented $\opt_i$ that is a planar graph that dominates $\hat{\candidate}_i$. One may apply Corollary~\ref{cor:size}, and Invariant~\ref{inv:cardinality} follows.
We begin by introducing the definitions that is required to formalize this argument.

Suppose $f$ is a non-triangulated face of $\hat{\candidate}_i$ with the boundary vertex sequence $\vertseq(f)$. Let $uv$ be an edge in $\opt_i$. We say that an edge $uv$ is a \emph{crossing edge} for any convex vertex $v_j \in \vertseq(f)$ if the edge $uv$ intersects the edges $v_{j-1}v_j$ and $v_jv_{j+1}$. We direct the crossing edge $uv$ from $u$ to $v$ if we first encounter the edge $v_{j-1}v_j$ as we move from $u$ to $v$. In this case, we refer to $v$ as the head and $u$ as the tail of this directed edge $\dir{uv}$.
Note that any such edge $uv \in \opt_i$ can be a crossing edge for many convex vertices in $\vertseq(f)$, and any convex vertex in $\vertseq(f)$ can have zero, one, or many crossing edges. However, from Lemma~\ref{lem:intersection}, we know that any edge $uv \in \opt_i$ that intersects a non-triangulated face $f$ has to be a crossing edge for some convex vertex $v_{} \in \vertseq(f)$. For any convex vertex $v_{j} \in \vertseq(f)$, let $\cs(v_{j})  $ be the set of crossing edges of $v_{j}$. We set $\cs(v_j)$ to be empty if $v_j$ is a reflex vertex. If $\cs(v_{j})$ is not empty, we define the furthest crossing edge of $v_{j}$  as the first edge of $\cs(v_j)$ that we encounter as we walk from $v_{j-1}$ to $v_j$. We can equivalently define the furthest crossing edge to be the last edge of $\cs(v_j)$ that we encounter as we walk from $v_j$ to $v_{j+1}$. This follows from the fact that $\opt_i$ is a planar graph, hence crossing edges cannot intersect. We make the following straight-forward observations for crossing edges which follows form Lemma~\ref{lem:intersection} and the fact that $\opt_i$ is a planar graph.
\begin{enumerate}
  \item[(i)] For any face $f \in F(\hat{\candidate}_i)$, let $v_j$ and $v_{j+1}$ be two consecutive vertices in $\vertseq(f)$. Let $S$ be the set of all edges of $\opt_i$\ that intersects $v_jv_{j+1}$, then $\cs(v_j)\cup\cs(v_{j+1})=S$.
  \item[(ii)] As we move from $v_j$ to $v_{j+1}$, first, we encounter all edges in $\cs(v_j)$ and only then  will we encounter  the edges of $\cs(v_{j+1})$. Therefore, as we move from $v_j$ to $v_{j+1}$, the furthest crossing edge of $v_j$  will appear immediately before the furthest crossing edge of $v_{j+1}$. 
\end{enumerate}

Next, consider any vertex $v_j \in \vertseq(f)$. We define two points $p_{2j}$ and $p_{2j+1}$  as follows. If $v_j$ does not have any crossing edges, we set $p_{2j}=p_{2j+1}=v_j$. Otherwise, if $v_j$ has a crossing edge, then let $e$ be the furthest crossing edge. We set $p_{2j}$ to be the point of intersection of $e$ with $v_{j-1}v_j$ and set $p_{2j+1}$ to be the point of intersection of  $v_jv_{j+1}$ with $e$. For any $j$, the following property is true for the segment $\overline{p_{2j}p_{2j+1}}$

\begin{itemize}
\item Suppose $v_j$ has a crossing edge. By construction, $\overline{p_{2j}p_{2j+1}}$ is contained in the furthest crossing edge of $v_j$ and the segment $\overline{p_{2j}p_{2j+1}}$ is contained inside the face $f$.
\end{itemize}

\noindent In addition, the following property is true for $\overline{p_{2j+1}p_{2j+2}}$. 
\begin{itemize}
\item By construction, $p_{2j+1}$ and $p_{2j+2}$ are points that lie on the edge $v_{j}v_{j+1}$. Furthermore, from property (ii), $\overline{p_{2j+1}p_{2j+2}}$ does not intersect with any other edge of $\opt_i$.
\end{itemize}

\noindent In the following lemma, we formally define the \emph{trace} of a non-triangulated face $f \in F(\hat{\candidate}_i)$ and non-triangulated face $f' \in F(\opt_i)$.
\begin{lemma}
\label{lem:trace}
For any face $f\in \hat{\candidate}_i$, as we walk along the cycle $C=\dir{p_0p_1},\dir{p_1p_2},\ldots,\dir{p_{2m-3}p_0}$, there is a unique non-triangulated face $f' \in F(\opt_i)$ that appears on the right with respect to the embedding of $\opt_i$. The face $f$ also appears on the right as we walk along $C$ with respect to the embedding of $\hat{\candidate}_i$. We define the region enclosed by the cycle $C$ as the \emph{trace} of $f$, and denote it by $\theta(f)=f\cap f'$ (Figure~\ref{fig:trace}).
\end{lemma}
\begin{proof}
We prove this claim by induction.

\textit{Base Case:}  There are two possibilities. Either (i) $p_0,p_1 = v_0$ or, (ii) $p_0$ and $p_1$ are intersection points of the furthest crossing edge $e_{0}$ of $v_0$ with edges $\dir{v_{m-1}v_0}$ and $\dir{v_0,v_1}$. In case (i), by construction $v_0$ is on the boundary of $f$. We set the face $f'$ to be the face of $\opt_i$ that lies to the right as we start to walk from $v_0$ towards $v_1$ along the edge $\dir{v_0v_1}$.
In case (ii), since $e$ is a crossing edge of $v_0$, $v_0$ is convex and\  $\overline{p_0p_1}$  lies completely inside the face $f$.  Therefore, $f$ will also appear on the right  as we walk along this edge. Also, since $\overline{p_0p_1}$ is contained inside $e$, as we walk from $p_0$ to $p_1,$ there is a unique face $f'$ that is to the right of $e$. 

\textit{Induction Step:} Suppose that the faces $f$ and $f'$ appear to the right of all edges starting from vertex $p_0$ to $p_{2j}$. We will now show that the claim is also true if we extend the path to $p_{2j+1}$ and then to $p_{2j+2}$. There are two possibilities: (i) $p_{2j}, p_{2j+1} = v_j$, or (ii) $p_{2j}$ and $p_{2j+1}$ are points of intersection of the furthest crossing edge $e_j$ of $v_j$.

In case (i), $p_{2j}, p_{2j+1} = v_j$, so the claim is trivially true for the path until $p_{2j+1}$. $\overline{p_{2j+1}p_{2j+2}}$ will be along the edge $v_{j}v_{j+1}$. Note that, as we walk along the path from $v_{j-1}$ through $v_j$ to $v_{j+1}$, there is no other edge of $f'$ incident on $v_j$ that appears to the right of this path.  Otherwise, such an edge will contradict Lemma~\ref{lem:intersection}. Therefore, the next line segment $\overline{p_{2j+1}p_{2j+2}}$ will continue to have $f'$ and $f$ on its right.

In case (ii), $p_{2j}$ and $p_{2j+1}$ are points of intersection of the furthest crossing edge $e_j$ of $v_j$. Since $v_j$ is a convex vertex, the entire segment $\overline{v_jv_{j+1}}$ is embedded to the right as we walk along from $v_{j-1}$ to $v_j$. Since $p_{2j}$ is an intersection point of $e_j$ with $v_jv_{j-1}$, and $p_{2j+1}$ is an intersection point of $e_j$ with $v_jv_{j+1}$, we will make a right turn at $p_{2j}$ and $p_{2j+1}$. When we make a right turn at intersection points, the faces on the right will continue to be on the right side. Therefore, the claim is true for the path until $p_{2j+2}$.   
\end{proof}

The following result also follows from the construction of the trace.

\begin{corollary}
\label{cor:unique-trace}
For any two distinct faces $f_{1}, f_2\in F(\hat{\candidate}_i)$, $\theta(f_1)$ and $\theta(f_2)$ are disjoint.
\end{corollary}
\begin{proof}
$f_1$ and $f_2$ are disjoint since $\hat{\candidate}_i$ is a $\PSLG$, therefore, their intersection with any face of $\opt_i$ will continue to be mutually disjoint.
\end{proof}

Lemma~\ref{lem:trace} together with Corollary~\ref{cor:unique-trace}) implies that the trace is a well-defined function mapping a non-triangulated face $f$ in $\hat{\candidate}_i$ to a single, unique, connected region. Therefore, we can augment $\opt_i$ to form a new graph, $\newopt$, such that the trace $\theta(f)$ corresponds to a face in $\newopt$, and that $\theta(f)$ is the dominating face for $f$. 

\begin{figure}
  \centering
  \includegraphics[width=0.8\textwidth]{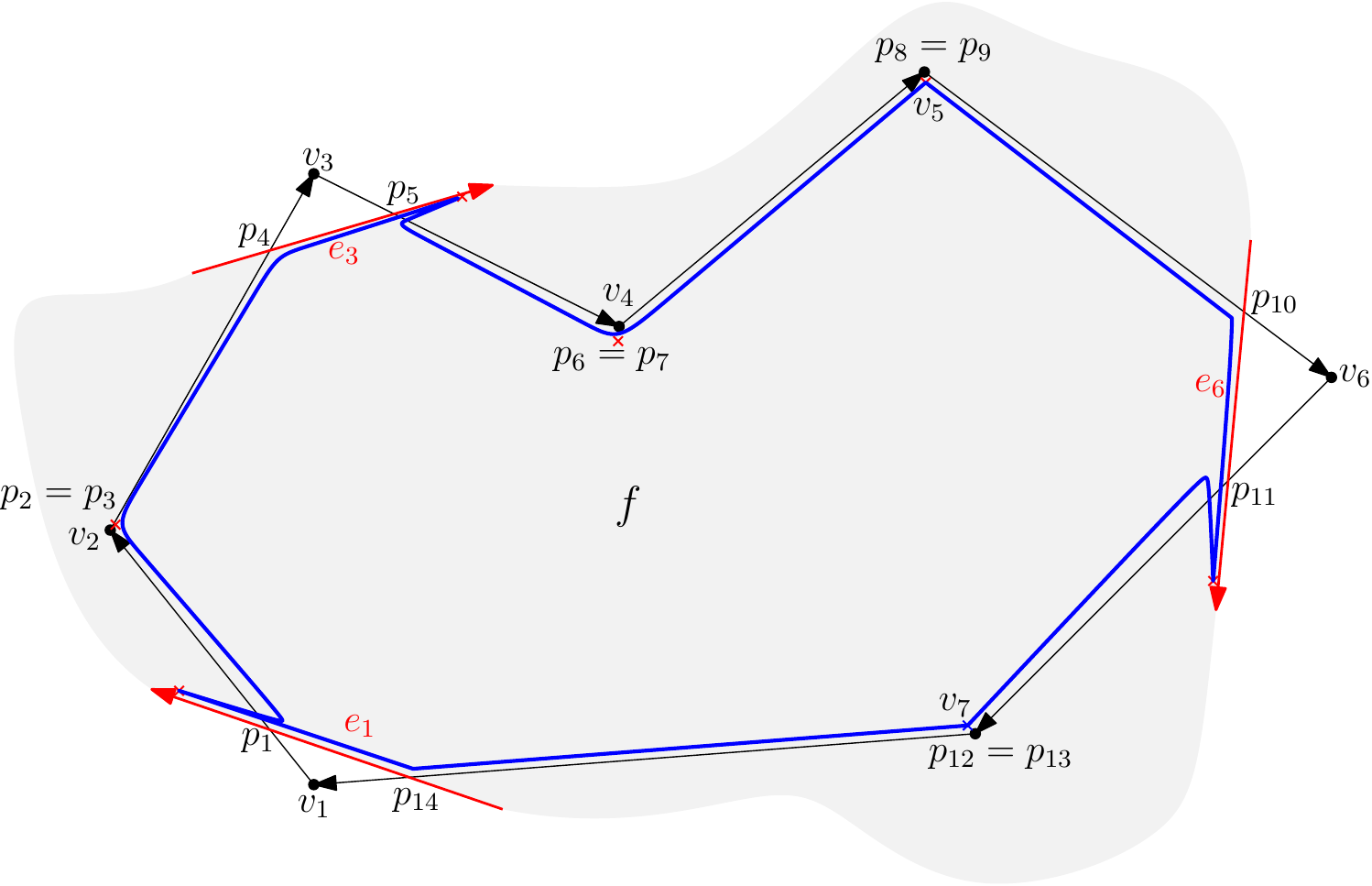}
  \caption[The trace, $\theta(f)$, of a non-triangulated face $f \in \hat{\candidate}_i$.]{The trace, $\theta(f)$, of a non-triangulated face $f \in \hat{\candidate}_i$. $e_1$, $e_3$, and $e_6$ are the furthest crossing edges for $v_1$, $v_3$ and $v_6$, respectively. Edges added in $\newedges{i}$ are depicted in blue.}\label{fig:trace}
\end{figure}

\paragraph{Adding edges to $\opt_i$} Next, we describe a procedure to embed a set of new edges $\newedges{i}$ to the straight-line embedding of $\opt_i$ to obtain a new graph $\newopt$ such that $\newopt$ remains a  planar graph. Additionally, for any face $f$ of $\hat\candidate_i$, the trace $\theta(f)$ is contained uniquely inside a face $f''$ of $\newopt$. Furthermore, $f''$ dominates $f$, so the signature $s(f'')\ge s(f)$. Note that the embedding of edges in $\newedges{i}$ that we construct is planar but not necessarily straight-line edges.

$\newedges{i}$ together with a planar embedding of $\newedges{i}$ is constructed as follows. Let $f \in F(\hat{\candidate}_i)$ have the vertex sequence $\vertseq(f) = \langle v_0, \ldots v_{m-1}\rangle$. Let the trace of $f$ be given by $\{p_0,\ldots, p_{2m-1}\}$ and let $e_k$ be the furthest crossing edge of the convex vertex $v_k$ for each $k$. Also, let $f'$ be the face of $\opt_i$ such that $\theta(f)=f\cap f' $ . Note that the edge $e_k$ exists only if $p_{2k-2} \neq p_{2k-1}$. Recollect  that each edge $e_k$ is an edge of a face $f'$ in $\opt_i$. We direct the edge $e_k$ so that as we walk along $e_k$ the face $f'$ appears on its right. We define a function $\psi$, which maps each index $j \in \{0, \ldots, m-1\}$ to a vertex $v$ on the boundary of $f$ or $f'$ as follows: if $e_j$ does not exist, then $\psi(j) = v_j$; otherwise, $\psi(j)$ is the vertex that is the head of the directed edge $e_j$.

Given this map $\psi$, we now describe a method for adding edges in $\newedges{i}$ that is added to $\newopt$. For every $\dir{v_jv_{j+1}}$, we add an edge $\dir{\psi(j)\psi(j+1)} \in \newedges{i}$ and embed this edge to construct a planar embedding of $\newopt$ as follows: Starting at $\psi(j)$, we draw the edge parallel and very close to $e_j$ until it reaches $p_{2j-1}$. At this point, we will draw the edge parallel and very close to the line segment $\overline{p_{2j-1}p_{2j}}$ until we reach the edge $e_{j+1}$. Finally, we will continue to draw the edge parallel and very close to the edge $e_{j+1}$ to $\psi({j+1})$.
Note that $\overline{p_{2j-1}p_{2j}}$ does not intersect any edge of $\opt_i$.
Note that edges $e_j$ and $e_{j+1}$ are edges of $\opt_i$, therefore, the newly added edge will not have any intersections with edges of $\opt_i$. If there are many vertices that have a single edge $e^*$ as a furthest crossing edge, there will be multiple edges that are drawn parallel to the edge $e^*$. We will draw all of them parallel and close to the edge $e^*$ (carefully stacked, one on top of the other, to avoid intersections).
  See Figure~\ref{fig:trace} for an example construction of the trace with the newly added edges, and Figure~\ref{fig:duplicate-edges} for an example where a single edge of $\opt_i$ is the furthest crossing edges for more than one vertex. This construction ensures that no two edges of $\newedges{i}$ intersect with each other and $\newopt$ is planar. By construction, there is a unique edge added for every $\dir{v_iv_{i+1}}$. Therefore, there is a unique face  $f''\in F(\newopt)$ that contains the trace of any face $f \in F(\hat\candidate_i)$ with $s(f)=s(f'')$. From the discussion above, the following lemma follows.

\begin{lemma}
The augmented graph $\newopt = \opt_i \cup \newedges{i}$ can be embedded without intersections. Furthermore, the face $f \in F(\hat\candidate_i)$ has the same signature as the face $f'' \in F(\newopt)$ that contains the trace of $f$.
\end{lemma}   

\begin{figure}
  \centering
  \begin{subfigure}{0.32\textwidth}
    \centering
    \includegraphics[width=0.98\textwidth]{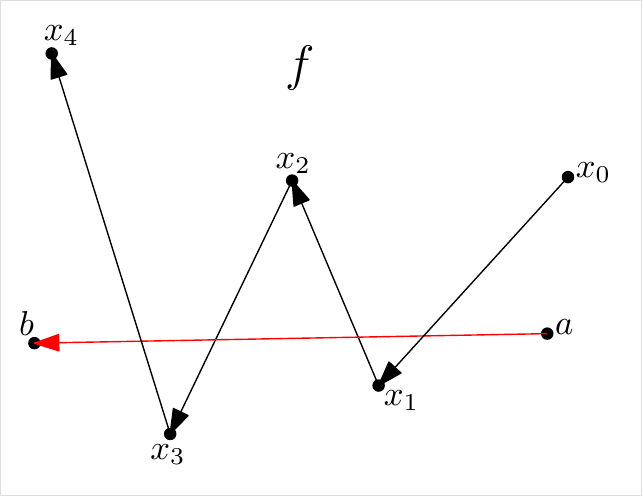}
    \caption{}
  \end{subfigure}
  \begin{subfigure}{0.32\textwidth}
    \centering
    \includegraphics[width=0.98\textwidth]{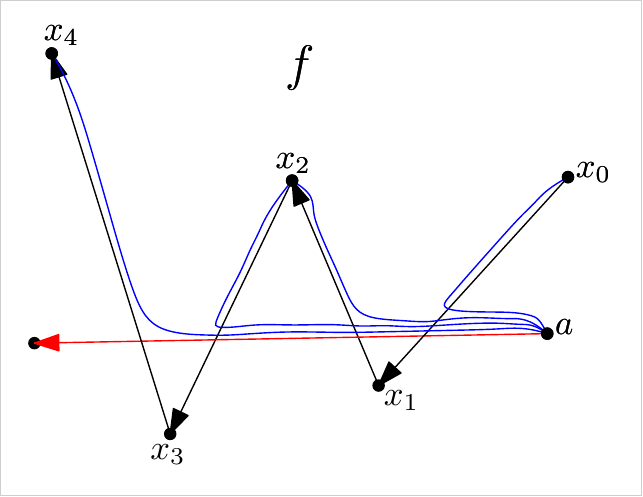}
    \caption{}
  \end{subfigure}
  \begin{subfigure}{0.32\textwidth}
    \centering
    \includegraphics[width=0.98\textwidth]{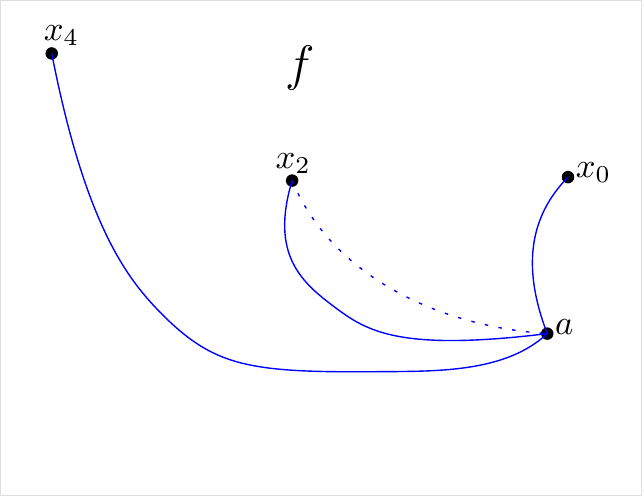}
    \caption{}
  \end{subfigure}
  \caption[Duplicate edges created in the construction of $\newedges{i}$.]{Duplicate edges created in the construction of $\newedges{i}$. (a) the edge $ab \in \opt_i$ intersects the face $f$ in $\hat{\candidate}_i$. The boundary vertex sequence of $f$ at this point is $\langle x_0, x_1, x_2, x_3, x_4 \rangle$ which contributes $4$ to $s(f)$. (b) Edges in $\newedges{i}$ are added and embedded as depicted in blue. (c) The edges of $\newedges{i}$. Note that there are multiple edges between $a$ and $x_2$, we show that we can delete one copy. This results on a potion on the boundary which also contributes $4$ to the signature.}
  \label{fig:duplicate-edges}
\end{figure}

However, we may create multiple copies of the same edges in $\newedges{i}\cup \opt_i$ (see Figure~\ref{fig:duplicate-edges}). To overcome this difficulty, we remove duplicate copies of the same edge. Removal of an edge from a planar graph will merge two faces $h$ and $h'$ into a single face $h''$ where $h''$ has a larger signature than either $h$ or $h'$. In Lemma~\ref{lem:embed}, we show that no edge with multiple copies participates in more than one dominating face. Therefore, removal of an edge only increases the signature of the dominating face it participates in, therefore, the face remains a dominating face.

\begin{figure}
  \centering
  \includegraphics[width=0.25\textwidth]{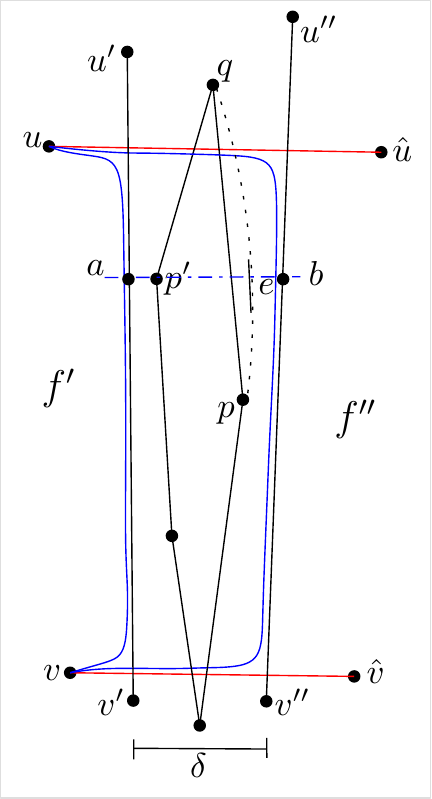}
  \caption[Trace between duplicate edges (Lemma~\ref{lem:embed}).]{There can be no trace between two duplicate edges in $\newedges{i}$ as shown in Lemma~\ref{lem:embed}.}\label{fig:unique-dominating-face}
\end{figure}
\begin{lemma}
\label{lem:embed}
For any edge $uv$ that has two or more copies in $\newopt$ between $u$ and $v$, no copy of $uv$ can participate in two distinct dominating faces. 
\end{lemma}
\begin{proof}
Let there be two copies of the edge $uv$ embedded in $\newopt$. We assume both these copies were in the set $\newedges{i}$. The case where only one edge
was added in $\newedges{i}$ has an identical argument. We assume that the two copies were added to $\newedges{i}$ due to two edges $u'v'$ and $u''v''$ of $\hat{\candidate}_i$ and that these edges support the trace of faces $f_1, f_2\in F(\hat{\candidate}_i)$. Let $uv$ be the edge that was added to $\newedges{i}$ for both $u'v'$ and $u''v''$. Note that $u'v'$ and $u''v''$ appear on the boundary of $f_1$ and $f_2$. Since the edge $uv$ was added for
both $u'v'$ and $u''v''$, from our construction, the edges $u'v'$ and $u''v''$ will intersect two edges $e_u = u\hat{u}$ and $e_v
= v\hat{v}$ that are on the boundary of some face $f'\in F(\opt_i)$ . See Figure~\ref{fig:unique-dominating-face} for the construction. Without loss of generality, assume $u\hat{u}$ is an edge parallel to the $x$-axis. To prove this lemma, it suffices if we show that there is no trace of any other face in $F(\hat{\candidate}_i)$ inside the region $R$ enclosed by $u'v'$ on the left, $u''v''$ on the right, $u\hat{u}$ from the top and $v\hat{v}$ from the bottom. Suppose, we can prove this claim, then deleting one copy of $uv$ (say the one corresponding to the trace of $f_2$) will only merge the dominating face $f_2'$ of $f_2$ in $\tilde{\opt_i}$ with a face that does not contain any trace. Hence, despite this deletion, we will continue to have a unique mapping of any face $f \in F(\hat{\candidate_i})$ to a dominating face of $\tilde{\opt_i}$. 

For the sake of contradiction, let there be another non-triangulated face $\tilde{f} \in \hat{\candidate}_i$ whose trace is in the region
$R$.  In this case, there must be at least two vertices of $\tilde{f}$ that lie inside $R$\ (otherwise either  $\tilde{f}$ will  be a triangle and $\theta(\tilde{f})$ is not defined or $\theta(\tilde{f})$ will not intersect $R$). Consider any such point $p'$ that is not the topmost or the bottommost vertex of $\tilde{f}$. Draw a horizontal line passing through $p'$. Let this line intersect $u'v'$ and $u''v''$ at $a$ and $b$ respectively. Since both $u\hat{u}$ and $v\hat{v}$ has a length at most $\delta = \frac{\rand 3^{i-1}}{\sqrt{2}}$ and since $u\hat{u}$ is horizontal, it follows that $\|ab\|\le \delta$. Without loss of generality, suppose we are inside the face $\tilde{f}$ as we begin to walk from $p'$ towards $b$. Let $pq$ be the first edge of $\tilde{f}$ that we encounter as we walk from $p'$ towards $a$ and let $q'$ be the intersection point of $p'a$ with $pq$. Clearly,  $p \neq p'$ and $q\neq p'$. By construction, the edge $pq$ is $\delta$-visible to $p'$ with respect to the face $\tilde{f}$.  Let $\tilde{f}'$ be the face of $\candidate_i$ which, after the execution of Phase 1 of the algorithm, created the face $\tilde{f}$. If we continue to walk from $q'$ towards $a$, we will encounter an edge $e$ of the chain $\chain(p,q)$ in $\tilde{f}'$. This follows from the fact that the chain $\chain(p,q)$ does not intersect with $u''v''$ and the region $E(p,q)$ contains no points of $P$. \ Therefore $e$  is  $\delta$-visible to $p'$.  From Lemma~\ref{lem:delta-visibility}, it follows that the end points of $e$ should belong to a triangulated chain of $\tilde{f'}$ that has  $p'$ as one of its end points, implying that either $p'=p$ or $p'=q$, leading to a contradiction.\  A similar argument extends to all other cases.
\end{proof}
   
It follows that $\newopt$ is a planar graph that dominates $\hat{\candidate}_i$. Thus, by Corollary~\ref{cor:size}, Invariant~\ref{inv:cardinality} follows. Since both Invariant~\ref{inv:criticallength} (after proving properties (P1)--(P4) in Section~\ref{subsec:proof-grid-props}) and Invariant~\ref{inv:cardinality} holds, this completes the proof of Theorem~\ref{thm:approx}.


\subsection{Proving Properties of a Maximal $\PSLG$ in $\adjgraph{i}$}
\label{subsec:proof-grid-props}

In the following section we prove (P1)--(P4) presented in Section~\ref{subsec:grid-properties}, which were used in the proof of the invariants. First recall that by Lemma~\ref{lem:2-or-3-cells-apart}, we know that vertices in $1$-chains are at most $3$ cells apart, and vertices participating in a $2$-chain are at most $4$ cells apart. This restricts the number of possible configurations of cells in which vertices participating in either type of chain can appear. We can deduce certain properties that arise due to the possible configurations, expressed in Lemma~\ref{lem:three-cell-ch} and Lemma~\ref{lem:four-cell-ch}, which will be used in the proof of the properties.

\begin{lemma}
\label{lem:three-cell-ch}
  Let $C_1$, $C_2$, and $C_3$ be three cells in $\grid{i}$ such that $C_1 \in N(C_2)$ and $C_2 \in N(C_3)$ (note $C_1$ and $C_3$ are $3$ cells apart). Let $H$ be the region enclosed by the convex hull of $C_1$, $C_2$ and $C_3$. Then,
  \begin{itemize}
    \item[(a)] if $C_3 \in N(C_1)$, then $H \subset \N(C_1) \cap \N(C_2) \cap \N(C_3)$,
    \item[(b)] otherwise $H \subset (\N(C_1) \cap \N(C_2)) \cup (\N(C_2) \cap \N(C_3))$.
  \end{itemize}
  In particular, in either case $H \subset \N(C_2)$.
\end{lemma}
\begin{proof}
  For both (a) and (b) it is not difficult to establish the result by considering all possible configurations of three cells which satisfy the necessary conditions. Let $C_1$, $C_2$, and $C_3$ be as above.
  
  Suppose (a) $C_3 \in N(C_1)$, so these three cells are mutually adjacent. Not counting symmetry, there are only three possible cell configurations that allow three cells to be mutually adjacent: either the three cells are equal (Figure~\ref{subfig:mutually-adj:all-same}), or exactly two of the cells are equal, and the third is adjacent (Figure~\ref{subfig:mutually-adj:two-same}), or all three cells are distinct (Figure~\ref{subfig:mutually-adj:three-distinct}). In all three cases, the convex hull $H$ of the three cells is contained in the neighborhood of each cell, i.e. $H \subset \N(C_1) \cap \N(C_2) \cap \N(C_3)$.
  
  Now suppose (b) $C_3 \notin N(C_1)$. Then there are exactly four possible configurations (not counting symmetry) to arrange the cells, each of which are depicted in Figure~\ref{fig:three-cell-configs}. In each of these cases it is not difficult to see that any point in $H$ must be contained in the neighborhood of at least two cells: $H \subset (\N(C_1) \cap \N(C_2)) \cup (\N(C_2) \cap \N(C_3))$. Both (a) and (b) immediately implies $H \subset \N(C_2)$.
  \begin{figure}
    \centering
    \begin{subfigure}{0.18\textwidth}
      \centering
      \includegraphics[width=0.8\textwidth]{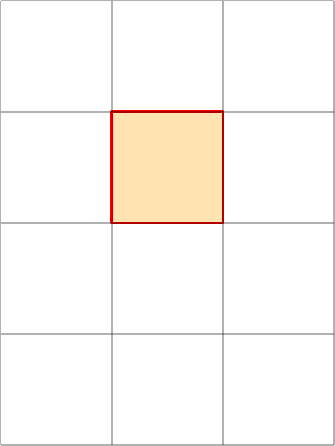}
      \caption{}
      \label{subfig:mutually-adj:all-same}
    \end{subfigure}
    \begin{subfigure}{0.18\textwidth}
      \centering
      \includegraphics[width=0.8\textwidth]{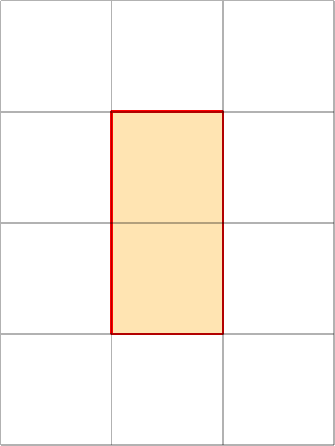}
      \caption{}
      \label{subfig:mutually-adj:two-same}
    \end{subfigure}
    \begin{subfigure}{0.23\textwidth}
      \centering
      \includegraphics[width=0.9\textwidth]{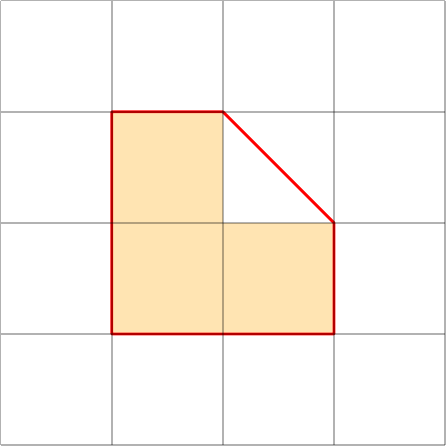}
      \caption{}
      \label{subfig:mutually-adj:three-distinct}
    \end{subfigure}
    \caption{Three mutually adjacent cells.}\label{fig:mutually-adjacent-cells}
  \end{figure}
  \begin{figure}
    \centering
    \begin{subfigure}{0.20\textwidth}
      \centering
      \includegraphics[width=0.75\textwidth]{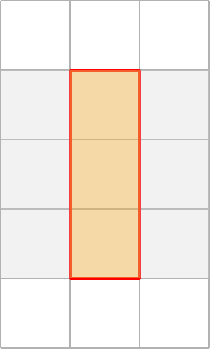}
      \caption{}
    \end{subfigure}
    \begin{subfigure}{0.22\textwidth}
      \centering
      \includegraphics[width=0.85\textwidth]{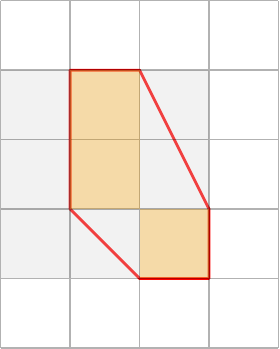}
      \caption{}
    \end{subfigure}
    \begin{subfigure}{0.22\textwidth}
      \centering
      \includegraphics[width=0.85\textwidth]{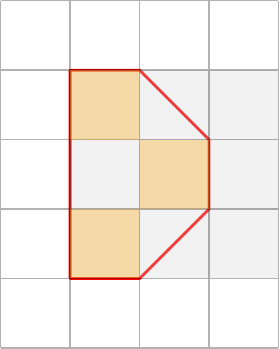}
      \caption{}
    \end{subfigure}
    \begin{subfigure}{0.27\textwidth}
      \centering
      \includegraphics[width=0.85\textwidth]{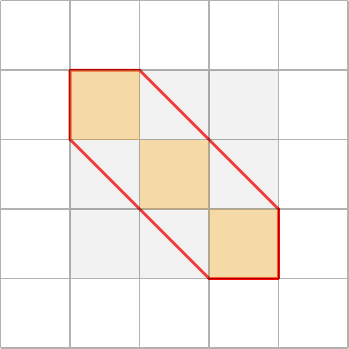}
      \caption{}
    \end{subfigure}
    \caption[Possible three cell apart cell configurations.]{Possible configurations for adjacent cells that are three cells apart.}\label{fig:three-cell-configs}
  \end{figure}
\end{proof}

\begin{lemma}
\label{lem:four-cell-ch}
  Let $C_1$, $C_2$, $C_3$, and $C_4$ be four cells in $\grid{i}$ such that $C_1 \in N(C_2)$, $C_2 \in N(C_3)$, $C_3 \in N(C_4)$ (so $C_1$ and $C_4$ are $4$ cells apart). Let $H$ be the region enclosed by the convex hull of $C_1$, $C_2$, $C_3$, and $C_4$. Then $H \subset (\N(C_2) \cup \N(C_3))$ in all but two cases.
\end{lemma}
\begin{proof}
  We establish this claim by considering all possible cell configurations. If the four cells are not distinct, this claim is reduced to Lemma~\ref{lem:three-cell-ch}. Therefore, consider only the twelve possible configurations (not counting symmetry) of four distinct cells satisfying these conditions depicted in Figure~\ref{fig:four-cell-configs}. Except for the two depicted in Figure~\ref{subfig:four-cell-prob1} and Figure~\ref{subfig:four-cell-prob2}, the claim holds for all other cases.
  \begin{figure}
    \centering
    \begin{subfigure}{0.16\textwidth}
      \centering
      \includegraphics[width=0.85\textwidth]{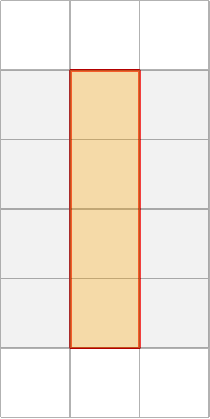}
      \caption{}
    \end{subfigure}\hspace{3mm}
    \begin{subfigure}{0.20\textwidth}
      \centering
      \includegraphics[width=0.85\textwidth]{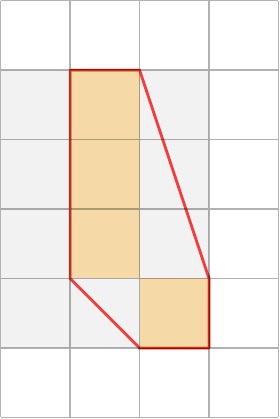}
      \caption{}
    \end{subfigure}\hspace{3mm}
    \begin{subfigure}{0.20\textwidth}
      \centering
      \includegraphics[width=0.85\textwidth]{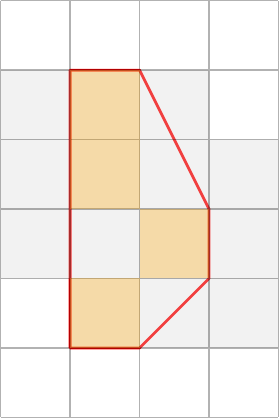}
      \caption{}
    \end{subfigure}\hspace{3mm}
    \begin{subfigure}{0.20\textwidth}
      \centering
      \includegraphics[width=0.85\textwidth]{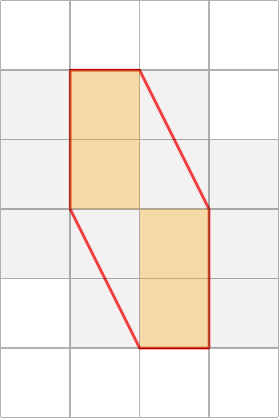}
      \caption{}
    \end{subfigure}\hspace{3mm}
    
    \begin{subfigure}{0.22\textwidth}
      \centering
      \includegraphics[width=0.85\textwidth]{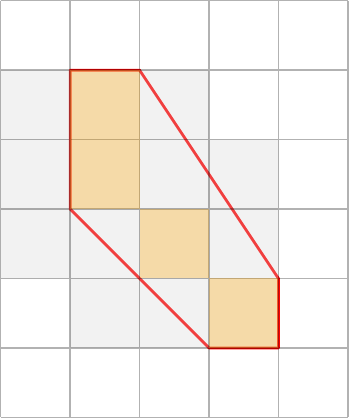}
      \caption{}
    \end{subfigure}
    \begin{subfigure}{0.27\textwidth}
      \centering
      \includegraphics[width=0.85\textwidth]{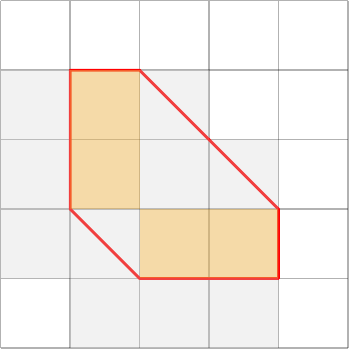}
      \caption{}
    \end{subfigure}
    \begin{subfigure}{0.20\textwidth}
      \centering
      \includegraphics[width=0.85\textwidth]{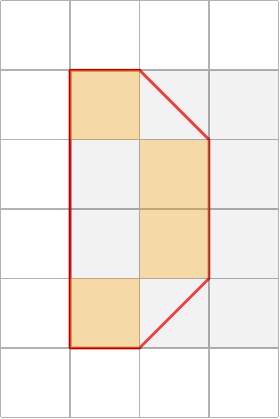}
      \caption{}
    \end{subfigure}
    \begin{subfigure}{0.22\textwidth}
      \centering
      \includegraphics[width=0.85\textwidth]{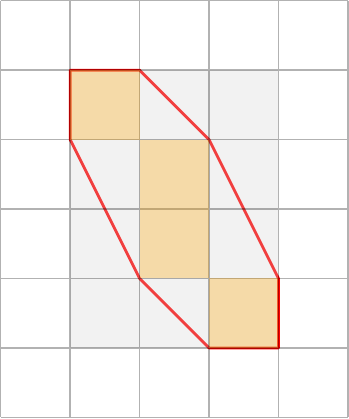}
      \caption{}
    \end{subfigure}
    
    \begin{subfigure}{0.20\textwidth}
      \centering
      \includegraphics[width=0.85\textwidth]{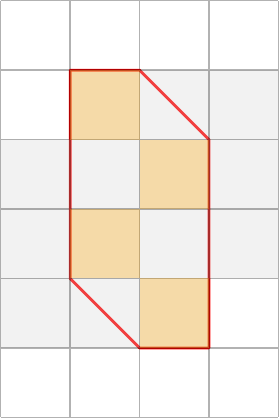}
      \caption{}
    \end{subfigure}
    \begin{subfigure}{0.23\textwidth}
      \centering
      \includegraphics[width=0.85\textwidth]{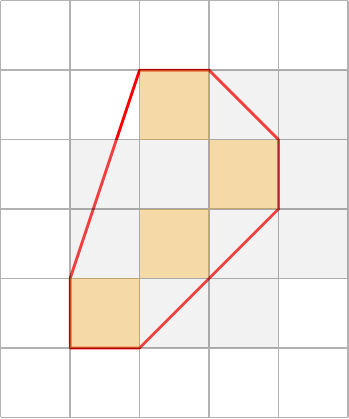}
      \caption{}
      \label{subfig:four-cell-prob1}
    \end{subfigure}
    \begin{subfigure}{0.24\textwidth}
      \centering
      \includegraphics[width=0.85\textwidth]{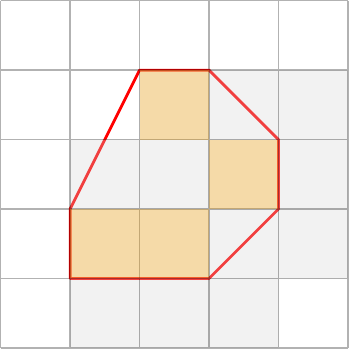}
      \caption{}
      \label{subfig:four-cell-prob2}
    \end{subfigure}
    \begin{subfigure}{0.24\textwidth}
      \centering
      \includegraphics[width=0.85\textwidth]{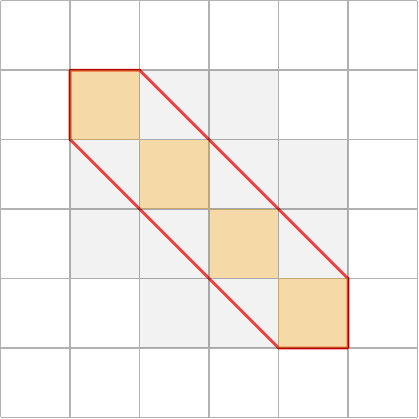}
      \caption{}
    \end{subfigure}
    \caption[Possible four cell apart cell configurations.]{Possible configurations for adjacent cells that are four cells apart.}\label{fig:four-cell-configs}
  \end{figure}
\end{proof}

We now describe a ``sweeping" procedure which will be re-used in subsequent proofs of the properties in order to show that a particular region is void of input points of $P$ other than those on the boundary of the region. Fix a vertex $v$, and an edge $xy$. We \emph{sweep} $\overline{vx}$, starting at $x$, along $xy$ towards $y$ if we define $\alpha(t) = x(1-t) + yt$ for $t \in [0, 1]$ and we consider the line segment $\overline{v\alpha(t)}$ for $t$ ranging from $0$ to $1$. We call $\overline{v\alpha(t)}$ the \emph{sweep line}. Therefore, one can imagine the line starting at $\overline{vx}$ for $t=0$ sweeping across a face by sliding $x$ along $\overline{xy}$, and reaching $\overline{vy}$ at time $t=1$. The starting and ending points need not be endpoints of an edge. If $p, q \in \overline{xy} \cup \{x, y\}$ is the starting and ending point, respectively, then one simply defines $\alpha(t)$ to be the appropriate parameterization of $\overline{pq} \subset \overline{xy}$. In this case, we say that, starting at $q$, we sweep $\overline{vp}$ along $xy$ towards $q$. 

We now proceed to prove (P1) through (P4). Let $\candidate$ be a maximal $\PSLG$ with respect to $\adjgraph{i}$. Let $f$ be a non-triangulated face in $\candidate$ with the boundary vertex sequence $\vertseq(f)$.

\begin{description}
  \item[(P1)]  Suppose $v_j \in \vertseq(f)$ is convex. Then $v_{j-1}$, $v_j$, and $v_{j+1}$ are in three distinct cells $C_{v_{j-1}}$, $C_{v_j}$, and $C_{v_{j+1}}$ of $\grid{i}$, respectively, and $C_{v_{j-1}} \notin N(C_{v_{j+1}})$.
\end{description}
\begin{proof}
  The fact that $v_{j-1}$, $v_j$, and $v_{j+1}$ appear consecutively in $\vertseq(f)$ implies that the edges $v_{j-1}v_j$ and $v_jv_{j+1}$ are in $\candidate$, hence $v_j \in \N(v_{j-1})$ and $v_j \in \N(v_{j+1})$. Suppose, for the sake of contradiction, that $v_{j-1} \in \N(v_{j+1})$, then $v_j$, $v_{j-1}$, and $v_{j+1}$ are in mutually neighboring cells which satisfy the conditions of Lemma~\ref{lem:three-cell-ch}(a). So $H$, the region enclosed by the convex hull of $C_{v_j-1}$, $C_{v_j}$, and $C_{j+k}$, is contained in $\N(v_{j-1}) \cap \N(v_j) \cap \N(v_{j+1})$. Since any point in the triangle formed by $v_{j-1}$, $v_j$, and $v_{j+1}$ must be contained in $H$, it follows that such a point is also contained in $\N(v_{j-1}) \cap \N(v_j) \cap \N(v_{j+1})$.
  Sweep the line segment $\overline{v_{j-1}v_j}$ along $v_jv_{j+1}$ towards $v_{j+1}$. Let $t^*$ be the smallest $t$ such that the sweep line $\overline{v_{j-1}\alpha(t^*)}$ intersects a vertex $v^*$. Then $v^*$ is visible to both $v_{j-1}$ and $v_j$, otherwise $v^*$ could not have been the first vertex $\overline{v_{j-1}\alpha(t)}$ intersects.
  Additionally, since $v^* \in H$, $v^*$ must be in $\N(v_j)$ and $\N(v_{j-1})$. Since $\candidate$ is maximal, it follows that both the edges $v_{j-1}v^*$ and $v_jv^*$ will be present in $\candidate$, hence $v_{j-1}$ and $v_j$ do not appear directly after one another in $\vertseq(f)$, contrary to assumption, so this case is impossible. If no such $v^*$ is encountered, then it implies that $v_{j+1}$ is visible to $v_{j-1}$, hence the edge $v_{j-1}v_{j+1}$ is present in $\candidate$, and $v_{j-1}$, $v_j$, and $v_{j+1}$ forms a triangulated face, contradicting the assumption that $v_{j-1}$, $v_j$, and $v_{j+1}$ appears in this order in $\vertseq(f)$. Therefore, one may conclude that $v_{j-1} \notin \N(v_{j+1})$, and hence $v_{j-1}$, $v_j$, and $v_{j+1}$ appears in three distinct cells, as claimed.
\end{proof}

\begin{description}
  \item[(P2)] Suppose $v_{j} \in \vertseq(f)$ and $v_{k}v_{k+1}$ is any edge on the boundary of $f$ such that $v_k \in \N(v_{j})$ (resp. $v_{k+1} \in \N(v_j)$) and $v_{k}v_{k+1}$  is visible to $v_{j}$ in $\vertseq(f)$. Then,
  \begin{itemize}
    \item[(i)] the chain $\chain$ from $v_j$ to $v_{k+1}$ (resp. chain $\tilde{\chain}$ from $v_k$ to $v_j$) in $\vertseq(f)$ is a $1$-chain, and,
    \item[(ii)] $v_{k+1}$ is a forward (resp. $v_k$ is backward) support vertex for  every vertex from $v_j$ to $v_{k-1}$ (resp. from $v_{k+2} $ to $v_j$).
  \end{itemize}
\end{description}
\begin{proof}
  We begin by showing (i), and (ii) follows. Suppose $v_k \in \N(v_j)$. Let $T$ be the triangle enclosed by the convex hull of the three points $v_j$, $v_k$, and $v_{k+1}$, which is contained in the convex hull of the three cells containing $v_j$, $v_k$, and $v_{k+1}$. Note that these three vertices satisfy the conditions of Lemma~\ref{lem:three-cell-ch}, therefore, $T \subset ((\N(v_j) \cap \N(v_k)) \cup (\N(v_k) \cap \N(v_{k+1}))$. Since $v_kv_{k+1}$ is visible to $v_j$, there exists a point $q \in \overline{v_kv_{k+1}}$ such that $q$ is visible to $v_j$, and the line segment $\overline{v_jq}$ does not intersect any edges of $\candidate$ and divides the triangle $T$ into two triangles, $T_1$ and $T_2$. Let $T_1$ be the triangle formed by $v_j$, $v_k$, and $q$, and let $T_2$ be the triangle formed by $v_j$, $q$, and $v_{k+1}$. We show that all input points in $T$ must be boundary vertices of $f$ which appear in $\vertseq(f)$.
  
  First we show that the triangle $T_2$ does not contain any input points in its interior. Starting at $q$, sweep the line $\overline{v_jq}$ along $v_kv_{k+1}$ towards $v_{k+1}$. Let $t^*$ the smallest $t$ such that the sweep line $\overline{v_j\alpha(t^*)}$ intersects an input point $u^*$; $u^*$ is contained in $T_2 \subset \N(v_k)$. Then either $v_k$ is visible to $u^*$, or there is an edge blocking $v_k$ from $u^*$ but an endpoint $u$ of such an edge must be contained in $T_1$ (otherwise the sweep line would have encountered it), and $u \in \N(v_k)$. Then either edge $v_ku^*$ or $uu^* \in \candidate$ since $\candidate$ is maximal, however, such an edge would intersect $\overline{v_jq}$, a contradiction. Therefore, the sweep line reaches $t=1$ without intersecting points, which implies that $T_2$ does not contain any input points in its interior, and therefore $v_j$ is visible to $v_{k+1}$ in $\vertseq(f)$. It follows that if all vertices from $v_{j+1}$ to $v_{k-1}$ is reflex, all these vertices would be visible to $v_{k+1}$, which would allow one to conclude $\chain$ is a $1$-chain.

  Now we show that the vertices between $v_j$ and $v_k$ must be reflex and also, $v_k$ is the only convex vertex between $v_j$ and $v_{k+1}$ in $\vertseq(f)$. Starting at $q$ sweep $\overline{v_jq}$ along $v_kv_{k+1}$ towards $v_k$. Let $t'$ be the smallest $t$ such that the sweep line $\overline{v_j\alpha(t')}$ intersects a vertex $w_0$. If  $w_0=v_{k}$, then $v_k$ is visible to $v_j$ and the edge $v_kv_j$ must exist in $\candidate$. This will mean that the vertex $v_k$ is convex and therefore the chain $\chain=\chain(v_j,v_{k+1})$ is a  $1$-chain. If  $w_0\neq v_k$, it follows that $w_0$ is visible to $v_j$. By Lemma~\ref{lem:three-cell-ch}, $w_0 \in (\N(v_j) \cap \N(v_k)) \cup (\N(v_k) \cap \N(v_{k+1}))$. We will show that  $w_0 \not\in \N(v_k) \cap \N(v_{k+1})$, since otherwise $w_0$ is visible to $v_{k+1}$ (by an identical sweeping argument as before) and the edge $w_0v_{k+1} \in \candidate$, a contradiction since $w_0v_{k+1}$ intersects $\overline{v_jq}$. Thus, $w_0 \in \N(v_j) \cap \N(v_k)$. Since $\candidate$ is maximal, $v_jw_0 \in \candidate$ and so $w_0$ is the first vertex after $v_j$ in $\vertseq(f)$, i.e., $w_0=v_{j+1}$. By construction, the edge $v_kv_{k+1}$ is also visible to $v_{j+1}$. Therefore, the   vertex that appears after  $v_{j+1}$ in $\vertseq(f)$ must lie to the left of the edge $v_jv_{j+1}$ implying $v_{j+1}$ is a reflex vertex.  Since all the conditions satisfied by $v_j$ is also satisfied by $v_{j+1}$, we can apply the same   sweeping arguments to show that $v_{k+1}$ is visible to $v_{j+1}$, either the next vertex of $v_{j+1}$ in $\vertseq(f)$ is  $v_k$ which is a convex vertex or the next vertex $v_{j+2}$ is a reflex vertex with $v_{j+2} \in \N(v_k)$ and $v_{j+2}$ has the edge $v_kv_{k+1}$ visible to it.   In this way, all the vertices between $v_j$ and $v_k$  can be shown to be reflex, contained in $\N(v_k)$ and $v_k$ is shown to be a convex vertex. Furthermore, every vertex between $v_j$ and $v_k$ is also visible to $v_{k+1}$. So, the chain $\chain(v_j,v_{k+1})$ is a $1$-chain, which shows (i). It now follows immediately that $v_k$ is the forward convex vertex for every vertex from $v_j$ to $v_{k-1}$, and therefore by definition, $v_{k+1}$ is the forward support vertex of these vertices. A symmetric argument holds for the other case, i.e., $v_{k+1} \in \N(v_j)$ (in which case $v_k$ will be the backwards support vertex of all vertices between $v_{k+1}$ and $v_j$ in $\vertseq(f)$).
\end{proof}

\begin{description}
  \item[(P3)] For any chain $\chain(v, y)$ from $v$ to $y$ in $\vertseq(f)$,   
  \begin{itemize}
    \item[(i)] if $\chain(v, y)$ is a $1$-chain with $v'$ as its only convex vertex, then the region $E = E(v, y)$ is contained in $\N(v')$, i.e., $E \subset \N(v')$, and $E$ contains no input points of $P$.  
    \item[(ii)] if the chain $\chain(v, y)$ is a $2$-chain with $v'$ and $v''$ as the two convex vertices, then the region $E = E(v, y)$ is such that $(E\cap(\N(v')\cup \N(v''))\cap P $ contains no points of $P$. In only two cases, $E\not\subset(\N(v')\cup \N(v''))$ (see Figure~\ref{fig:p3-four-cell-exception} and Figures~\ref{subfig:four-cell-prob1} and~\ref{subfig:four-cell-prob2}) and may contain  points of $P$. In all other cases, $E$ contains  no  points of $P$.
  \end{itemize}
\end{description}

\begin{proof} 
  We first show (i). Suppose $\chain(v, y) = \chain$ is a $1$-chain with $v'$ as its only convex vertex. Let $v_{prev}$ and $v_{next}$ be the vertex that appears before and after $v'$ in $\vertseq(f) $ . By (P1), we know that $v_{prev} \not\in \N(v_{next})$. Let $E = E(v, y)$.  Since $\chain$ is a $1$-chain, it follows that $v$ is visible to $y$. Let $H$ be the region enclosed by the convex hull of the three cells containing $v$, $v'$, and $y$. By definition of a $1$-chain, these three cells satisfy the conditions for Lemma~\ref{lem:three-cell-ch}, and hence $H \subset \N(v')$. Since $E \subset H$, it follows that $E \subset \N(v')$. We partition $E$ into three regions $E_1,E_2$ and $E_3$ as follows (See Figure~\ref{fig:p3proof}(i)).  Consider a ray $r$ that starts from $v'$ in the direction going towards $v_{prev}$ and another ray $r'$ from $v'$ going towards $v_{next}$. Let the $r$ and $r'$ intersect $vy$ at $y'$ and $y''$ respectively. Then, we set the region $E_1$ to be the triangle formed by $y', v'$ and $y''$. We set $E_2$ to be a region bounded by segments $\overline{v_{prev}y'}$ and $\overline{y'v}$ on the two sides and the reflex chain $\chain(v,v_{prev})$ on the third side. Similarly, we set $E_3$ to be the region bounded by segments $\overline{v_{next}y''}$ and $\overline{y''y}$ and the chain $\chain(v_{next},y)$  on the third side. If $v_{prev} =v$ (resp. $v_{next}=y$), then the region $E_2$(resp. $E_3$) will be an empty region.  We claim that none of the  three regions will  contain any points of $P$. 

Note that $E_1$ is a triangle. If $E_1$ contains a point $p \in P$ then $p$ will be visible to $v'$ and also in the neighborhood of $v'$ (since $E_1 \subset E \subset \N(v')$). Since $p$ does not participating in the chain $\chain$, there is no edge from between $p$ and $v'$. Therefore, we can add the edge from $p$  to $v'$ contradicting the fact that $\candidate$ is a maximal $\PSLG$.  \ 

The arguments for regions $E_2$ and $E_3$ are symmetric.  We will present the proof for  $E_2$. $E_2$ is a region bounded by the reflex chain $\chain(v, v_{prev})$ on one side and the segments $\overline{v_{prev}y'}$ and $\overline{y'v}$ on the other two sides.  By construction, $\overline{v_{prev}y'} \subset \overline{v'y'}$.  Since $v' \in \N(y)$,  $v_{prev} \not\in \N(y)$ and  $\N(y)$ is convex, it follows that the entire segment $\overline{v_{prev}y'}$ lies outside $\N(y)$. Since $y \in \N(y)$ and $y' \not\in \N(y)$, it follows that $\overline{y'v}$ also lies outside $\N(y)$. Note that for any point $p$ on an edge of the reflex chain $\chain(v,v_{prev})$, $\overline{yp}$ intersects $v_{prev}y'$.  Since $y \in \N(y), y' \not\in \N(y)$ and from convexity of $\N(y)$ it follows that the reflex chain $\chain(v, v_{prev})$ also does not intersect $\N(y)$.  Consequently the region $E_2$   does not intersect $\N(y)$.
 From Lemma~\ref{lem:three-cell-ch}(b), therefore, $E_2 \subset \N(v')\cap \N(v)$. It is easy to see that any four points $a,b,c,d$ such that  $a,b \in \N(c)\cap \N(d)$. Then, $a$ and $b$ will have an edge in the adjacency graph. Any point $x$ of $P$ inside $E_2$ will be visible to some vertex $q$ in the reflex chain $\chain(v,v_{prev})$.  Since every point of $\chain(v,v_{prev})$ including $q$ is in $\N(v)\cap\N(v')$ and $x$ is also inside $\N(v)\cap\N(v')$, there is an edge between $x$ and $q$ in the adjacency graph.\  Therefore, the edge $xq$ can be added to  $\candidate$ leading to a contradiction that $\candidate$  is a  maximal $\PSLG$. 

\begin{figure}
  \centering
  \includegraphics[height=6cm]{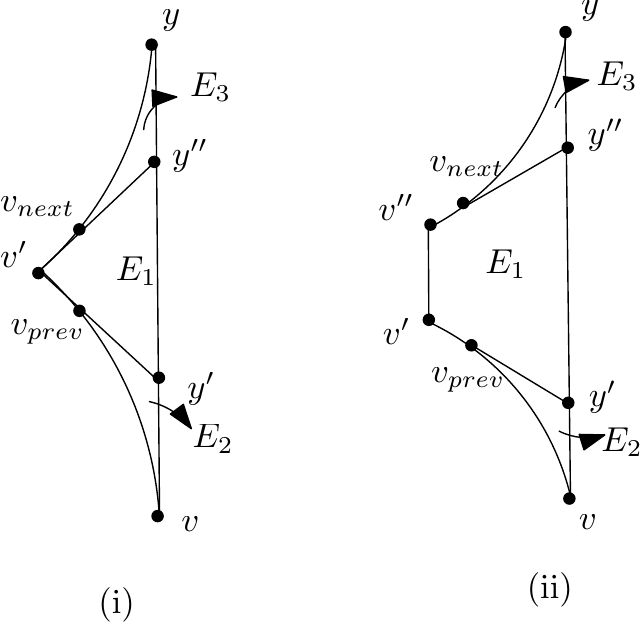}
  \caption[$E$ is divided into $E_1, E_2$ and $E_3$ for proof of (P3)(i) and (ii).]{Proof of (P3)(i) and (ii)}\label{fig:p3proof}
\end{figure}

   To show (ii), suppose $\chain$ is a $2$-chain such that $v'$ and $v''$ are its two convex vertices. Let $H$ be the region enclosed by the convex hull of the four cells containing $v$, $v'$, $v''$, and $y$. Then these four cells satisfy the conditions for Lemma~\ref{lem:four-cell-ch}, so for all but two cases, $H \subset (\N(v') \cup \N(v''))$. Then $E \subset H \subset (\N(v') \cup \N(v''))$, except for two cases. Let $v_{prev}$ be the vertex that appears before $v'$ and $v_{next}$ be the vertex that appears after $v''$ in $\vertseq(f)$. We can partition $E$ into three regions, namely $E_1$, $E_2$ and $E_3$ as follows (See Figure~\ref{fig:p3proof}(ii)). Draw a ray $r_1$ from $v'$ going towards $v_{prev}$. Let $y'$ be the intersection point of this ray with $\overline{yv}$. Draw another ray $r_2$\ from $v''$ going towards $v_{next}$ and let $y''$ be the intersection point of this ray with $\overline{yv}$. $E_1$ is the quadrilateral formed by the segments $\overline{v'v''}$, $\overline{v'y'}$, $\overline{v''y''}$ and $\overline{y'y''}$. $E_2$ is the region bounded by the chain $\chain(v, v_{prev})$ and segments $\overline{v_{prev}y'}$ and $\overline{y'v}$ whereas $E_3$ is formed by the chain $\chain(v_{next},y)$ and edges $\overline{v_{next}y''}$, $\overline{y''y}$. Since $\chain(v',y)$ is a $1$-chain, we can use an  identical argument to the case of $1$-chain to show that $E_2$ is empty. Similarly, since $\chain(v,v'')$ is a $1$-chain, we can use an identical argument to the case of $1$-chain to show that $E_3$ is empty. For $E_1$,  other  than those two exceptions, $E_1$ is contained inside $\N(v')\cup \N(v'')$ and therefore any point inside $E_1$ will also have an edge to $v'$ or $v''$ leading to a contradiction. Hence, except for the two cases, $E$ will not have any points of $P$ inside it. 
\end{proof}

\begin{description}
  \item[(P4)] For any vertex $v \in \vertseq(f)$, if an edge $xy$ on the boundary of $f$ is $\delta$-visible for $\delta=\frac{\rand3^{i-1}}{\sqrt{2}}$, then exactly one of $x$ and $y$ are in $\N(v)$ and the other is not.
\end{description}
\begin{proof}
  Let $C$ be the cell in grid $\grid{i}$ that contains $v$. Then let $J = \bigcup_{p \in C} \{z \in \reals^2: \|pz\| \leq \delta \}$ (see the red region in Figure~\ref{fig:cell-delta-ball}) be the Minkowski sum of cell $C$ with a ball of radius $\delta$. Since $xy$ is $\delta$-visible to $v$, there exists a point $q \in \overline{xy}$, such that $q$ is visible to $v$, and $q \in J$. Suppose, for the sake of contradiction, that (a) $x, y \in \N(v)$ or (b) $x, y \notin \N(v)$. We consider the cases separately.
  
  (a) $x, y \in \N(v)$. Without loss of generality, assume $v$, $x$, and $y$ appear in this order in $\vertseq(f)$ (a symmetric argument holds if the order is $y$, $x$, $v$). Since $q$ is visible to $v$ in $\vertseq(f)$, the line segment $\overline{vq}$ does not intersect any edges of $\candidate$.
  Starting at $q$, sweep $\overline{v\alpha(t)}$ along $xy$ towards $x$. Let $t^*$ be the smallest $t$ such that the sweep line $\overline{v\alpha(t^*)}$ intersects an input point $v^*$; if the sweep line intersects no input point, let $v^* = x$. Then $v^*$ is visible to $v$ so that $vv^* \in \candidate_i$. Similarly, sweep $\overline{vq}$ towards $y$, starting at $q$ and moving from $q$ along $xy$ towards $y$. Let $t'$ be the smallest $t$ such that this sweep line intersects a vertex $v'$, or if the sweep line intersects no vertex, set $v' = y$. Then $v$ and $v'$ are visible to each other, hence $vv' \in \candidate$. In addition, $v'$ and $v^*$ are in each other's neighborhoods so that $v'v^* \in \candidate$. However, this implies $v$, $v'$, and $v^*$ forms a triangulated face and $\overline{vq}$ is not contained in $f$, contradicting the fact that $q$ is visible to $v$. Therefore, this case is impossible.

\begin{figure}
  \centering
  \includegraphics[width=0.28\textwidth]{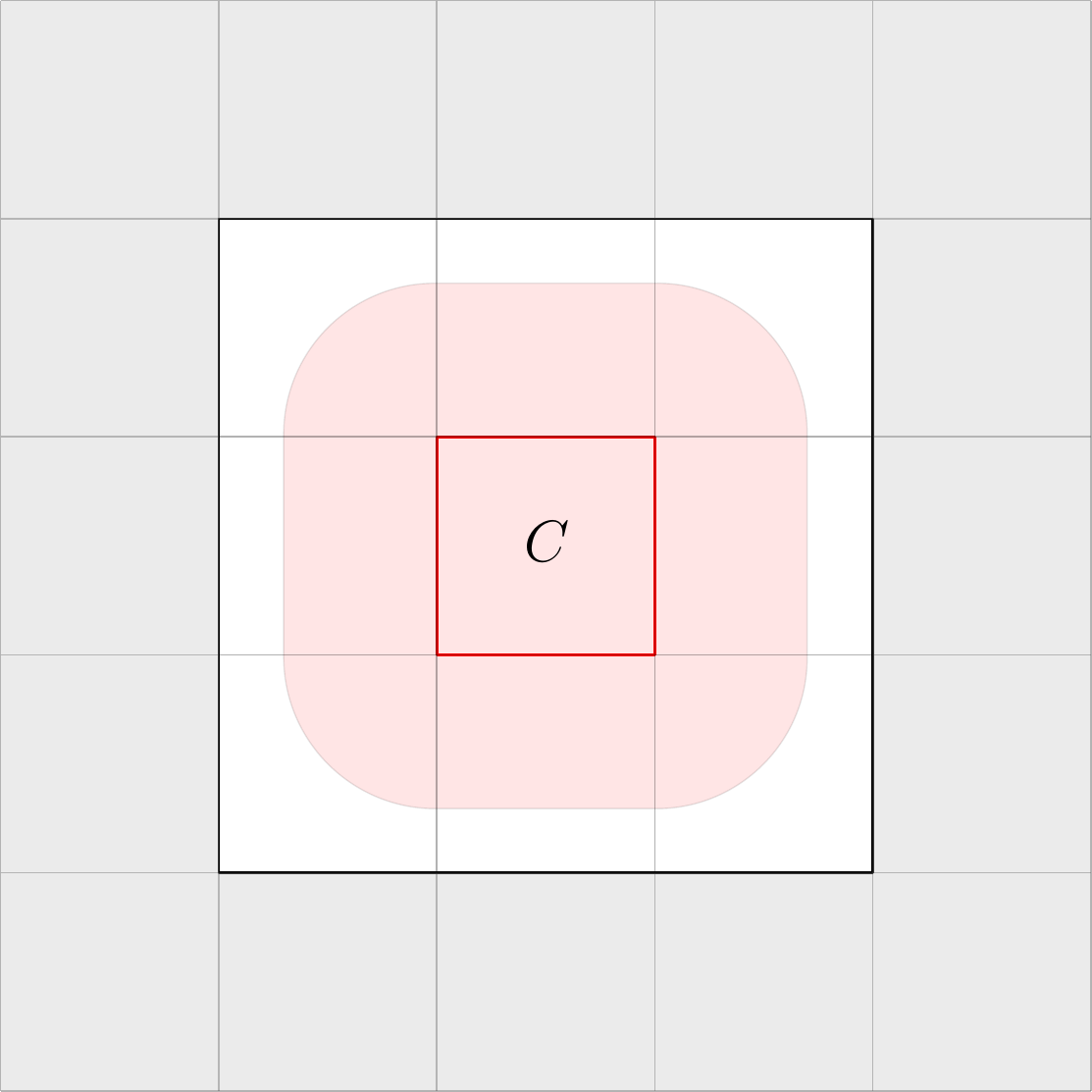}
  \caption[The set $J$ of $C$ (proof of (P4)).]{The set $J$ (depicted by the red shaded region), for the cell $C$.}\label{fig:cell-delta-ball}
\end{figure}
  
  (b) $x, y \notin \N(v)$. Note that since $xy \in \candidate \subset \adjgraph{i}$, $x$ and $y$ must be in neighboring cells. Furthermore, the cells containing $x$ and $y$ must be adjacent to $N(v)$ (the grey region in Figure~\ref{fig:cell-delta-ball}). It is not difficult to see that an edge between $x$ in $y$ in two neighboring grey cells in Figure~\ref{fig:cell-delta-ball} cannot intersect $J$. The result follows.
\end{proof}


\section{Extensions}
\label{sec:extension}
In this section, we extend the analysis of our algorithm to the $q$-MWT problem for $q \ge 2$. After that, we will also show that by making a small modification to our algorithm, we can improve the worst-case bound on the approximation ratio to $21$
and the bound on the expected approximation ratio to $14$. 
\paragraph{Extending to the $q$-MWT problem.} For the $q$-MWT problem, let $\opt_q$ be the optimal triangulation. We will show that the triangulation produced by our algorithm is also a worst-case $24$-approximation and a expected $16$-approximation of the optimal $q$-MWT.    As already shown, Invariant 1 holds for our algorithm. We can also show that a slightly modified Invariant 2 also holds.      

\noindent{\bf Modified Invariant 2:}
$|\hat{\candidate}_i| \geq |\opt_{q,i}|$, where $\opt_{q,i}$ is a restricted optimal triangulation containing only edges of $\opt_q$ whose length is less than $\frac{\rand 3^{i-1}}{\sqrt{2}}$.

 The proof of second invariant is based on the crucial properties established by Lemma~\ref{lem:delta-visibility}, Lemma~\ref{lem:intersection} and Corollary~\ref{cor:components}. All these Lemmas hold for any $\PSLG$ where the Euclidean distance between the end-points of its edges is no more than $\frac{\rand 3^{i-1}}{\sqrt{2}}$. This includes the $\PSLG$  $\opt_{q,i}$ and hence Lemma~\ref{lem:delta-visibility},~\ref{lem:intersection} and Corollary~\ref{cor:components}  is true for  $\opt_{q,i}$. Therefore, we can use the arguments identical to those in Section~\ref{sec:inv2} and prove the modified version of Invariant 2.

 The proof of approximation ratio is very similar to the proof of Theorem~\ref{thm:approx}. For the bottleneck triangulation problem, i.e., $q=\infty$, we can use arguments identical to  proof of Theorem~\ref{thm:approx}, and bound the ratio $\alpha_j$ of the $j^{th}$ edge of $\tau$ and $\tau^*$ by $\alpha_j \le 24$  and $\mathbb{E}[\alpha_j]\le 16$. We can invoke this bound for the last edge of $\tau$ and $\tau^*$ and bound the approximation ratio for the bottleneck triangulation.

For any finite integer  $q$, we can define the approximation ratio to be $$\alpha = w(\candidate)/ w(\opt_{q})=(w'(\candidate)/w'(\opt_q))^{1/q};$$ here, for any edge $e=uv$, $w'(e)=\|uv\|^q$ and for any triangulation $\triangulation$, $w'(\triangulation)=\sum_{uv \in \triangulation}\|uv\|^q$. Let $\alpha'_j = w'(a_j)/w'(t_j)$ and $\alpha' = w'(\candidate)/w'(\opt_q)$  . Similar to Equation~\ref{eq:approx}, we get

\begin{eqnarray}
\alpha' = \frac{\wt'(\candidate)}{\wt'(\opt)} &=&
\sum_{j=1}^m\frac{\wt'(t_j)}{\wt'(\opt)}\cdot \frac{\wt'(a_j)}{\wt'(t_j)}
= \sum_{j=1}^m\beta_j\cdot \frac{\wt'(a_j)}{\wt'(t_j)},
\label{eq:approx1}
\end{eqnarray}
where $\beta_j = \frac{\wt'(t_j)}{\wt'(\opt)}$. Since $\beta_j > 0$ and
$\sum_{j=1}^m \beta_j =1$,  $\alpha'$ is a weighted average of all
$\alpha'_j$ values. Using the invariants as in the proof of Theorem~\ref{thm:approx}, we can bound \[ \alpha'_j = \dfrac{\wt'(a_j)}{\wt'(t_j)} \leq
   \dfrac{(\rand 4\sqrt{2} \cdot 3^{k-1})^q}{\left(\dfrac{\rand3^{k-2}}{\sqrt{2}}\right)^q} = 24^q. \]  
Applying this bound to  (\ref{eq:approx1}), we can bound $\alpha'$ by $24^q$, i.e., $\alpha' \le 24^q$. Since the approximation ratio $\alpha = \alpha'^{1/q}$, we can bound the approximation ratio $\alpha$ by $24$ in the worst case. Similar to the proof of Theorem~\ref{thm:approx}, we can bound the expected value of $\alpha'_j$ by $16^q$ and using linearity of expectation, we can bound the expected value of $\alpha'$, $\mathbb{E}[\alpha']\le 16^q$. Since $f(x)=x^{1/q}$ is a concave function, from Jensen's inequality, we get $\mathbb{E}[\alpha]=\mathbb{E}[(\alpha')^{1/q}] \le (\mathbb{E}[\alpha'])^{1/q}$. Therefore, $\mathbb{E}[\alpha] \le 16$. 
\paragraph{Improving approximation ratio.} The approximation ratio can be improved to $21$ (and the expected
approximation ratio to $14$) by making the following simple
modification to our algorithm (cf. Section~\ref{sec:algorithm}) based
on the following observation. For any triangulated chain $\chain(v_{j+k},
v_p)$ generated by the algorithm, $\chain(v_{j+k}, v_p)$ is a
$2$-chain only if the algorithm executes step 1(a) or 1(c). Using the same notation as in the presentation of the algorithm, consider
the case where $\chain(v_{j+k}, v_p)$ is generated in 1(a) (so $v_p =
v_s$). We have shown in Lemma~\ref{lem:triangulated-chains-1-2-chains}
that edges in $\edges(v_{j+k}, v_s)$ that triangulate $E(v_{j+k},
v_s)$ are bounded by $4\sqrt{2}\cdot \rand 3^{i-1}$. However, by
construction $\chain(v_{j+k}, v_l)$ and $\chain(v_{l+1}, v_s)$ are
both $1$-chains, so if both chains could connect to their
support vertices, edges in $\edges(v_{j+k}, v_s)$ would be bounded by
$3\sqrt{2}\cdot \rand 3^{i-1}$. Since edges added for one chain hides
the support of the other, both sets of edges cannot be added without violating planarity, so there is a conflict. The algorithm resolves
this conflict by always connecting edges for $\chain(v_{j+k},
v_{l-2})$ to its forward support and adding longer edges (with endpoints $4$ cells apart) by adding edges from $\chain(v_{l+1}, v_s)$ to $v_{j+k}$. Equivalently, one can add the longer edges from
$\chain(v_{j+k}, v_{l-2})$ to $v_s$ and add edges from
$\chain(v_{l+1}, v_s)$ to its backward support $v_{j+k}$. By adding a
check in the algorithm to make the choice which results in the smaller weight for edges in $\edges(v_{j+k}, v_s)$, we can make the argument that the upper bound on
the average length of an edge added to $\hat{\candidate}_i$ is
$\frac{7}{2}\sqrt{2}\cdot \rand 3^{i-1}$. This helps improve the ratio to $21$ (or expected $14$).

\section{Conclusion}
\label{sec:conclusion}

We introduced a polynomial time approximation algorithm that computes a triangulation for approximating  $q$-MWT
for every $q\ge 1 $ including the case of minimum weight triangulation and the minimax length triangulation with a worst-case approximation ratio of $21$, and an expected
approximation ratio of $14$. This is achieved by partitioning edges
into levels using grids, and applying a combination of the ring heuristic and the greedy heuristic
at each level $i$ to obtain a partial candidate solution.

It is an open question whether the techniques developed here can be adapted in
order to design a Polynomial Time Approximation Scheme (PTAS) for the
MWT. Any such construction will maintain finer grids, multiple candidate
solutions at each grid level, and use dynamic programming to mimic the restricted optimal solution for each grid.

\bibliography{mwt_approx_focs2017}

\begin{thebibliography}{10}

\bibitem{Arora-98}
S.~Arora.
\newblock Polynomial time approximation schemes for euclidean traveling
  salesman and other geometric problems.
\newblock {\em J. ACM}, 45(5):753--782, 1998.

\bibitem{BerEpp-CEG-92}
M.~W. Bern and D.~Eppstein.
\newblock {Mesh generation and optimal triangulation}.
\newblock In Ding-Zhu Du and Frank Kwang-Ming Hwang, editors, {\em Computing in
  Euclidean Geometry}, number~1 in Lecture Notes Series on Computing, pages
  23--90. World Scientific, 1992.

\bibitem{Clarkson-91}
K.~L. Clarkson.
\newblock Approximation algorithms for planar traveling salesman tours and
  minimum-length triangulations.
\newblock In {\em Proceedings of the Second Annual ACM-SIAM Symposium on
  Discrete Algorithms}, SODA '91, pages 17--23, Philadelphia, PA, USA, 1991.
  Society for Industrial and Applied Mathematics.

\bibitem{Dantzig-85}
G.~B. Dantzig, A.~J. Hoffman, and T.~C. Hu.
\newblock Triangulations (tilings) and certain block triangular matrices.
\newblock {\em Mathematical Programming}, 31(1):1--14, 1985.

\bibitem{DupGott-70}
R.~D. Düppe and H.~J. Gottschalk.
\newblock Automatische interpolation von isolinien bei willkürlich verteilten
  stützpunkten.
\newblock {\em Allgemeine Vermessungs-Nachrichten}, 77:423--426, 1970.

\bibitem{TanEde-93}
Herbert Edelsbrunner and Tiow~Seng Tan.
\newblock A quadratic time algorithm for the minmax length triangulation.
\newblock {\em SIAM Journal on Computing}, 22(3):527--551, 1993.

\bibitem{Eppstein-94}
D.~Eppstein.
\newblock Approximating the minimum weight steiner triangulation.
\newblock {\em Discrete Comput. Geom.}, 11(2):163--191, 1994.

\bibitem{GareyJohnson-79}
M.~R. Garey and D.~S. Johnson.
\newblock {\em Computers and Intractability: A Guide to the Theory of
  {NP}--Completeness}.
\newblock W. H. Freeman \& Co., New York, NY, USA, 1979.

\bibitem{LevKrz-98}
C.~Levcopoulos and D.~Krznaric.
\newblock Quasi-greedy triangulations approximating the minimum weight
  triangulation.
\newblock {\em J. Algorithms}, 27(2):303--338, May 1998.

\bibitem{Lloyd-77}
E.~L. Lloyd.
\newblock On triangulations of a set of points in the plane.
\newblock SFCS '77, pages 228--240, Washington, DC, USA, 1977. IEEE Computer
  Society.

\bibitem{ManZob-79}
G.~K. Manacher and A.~L. Zobrist.
\newblock Neither the greedy nor the {D}elaunay triangulation of a planar point
  set approximates the optimal triangulation.
\newblock {\em Inform. Process. Lett.}, 9(1):31--34, 1979.

\bibitem{Mitch-96}
J.~S.~B. Mitchell.
\newblock Guillotine subdivisions approximate polygonal subdivisions: A simple
  new method for the geometric k-mst problem.
\newblock In {\em Proceedings of the Seventh Annual ACM-SIAM Symposium on
  Discrete Algorithms}, SODA '96, pages 402--408, Philadelphia, PA, USA, 1996.
  Society for Industrial and Applied Mathematics.

\bibitem{MulRote-06}
W.~Mulzer and G.~Rote.
\newblock Minimum-weight triangulation is {NP}-hard.
\newblock {\em CoRR}, abs/cs/0601002, 2006.

\bibitem{PlaHong-87}
D.~A. Plaisted and J.~R. Hong.
\newblock A heuristic triangulation algorithm.
\newblock {\em J. Algorithms}, 8(3):405--437, 1987.

\bibitem{RemySteg-09}
J.~Remy and A.~Steger.
\newblock A quasi-polynomial time approximation scheme for minimum weight
  triangulation.
\newblock {\em J. ACM}, 56(3):15:1--15:47, May 2009.

\bibitem{ShamosHoey-75}
M.~I. Shamos and D.~Hoey.
\newblock Closest-point problems.
\newblock In {\em Proceedings of the 16th Annual Symposium on Foundations of
  Computer Science}, pages 151--162, Washington, DC, USA, 1975. IEEE Computer
  Society.

\bibitem{YoYo-11}
A.~Yousefi and N.~E. Young.
\newblock On a linear program for minimum-weight triangulation.
\newblock {\em CoRR}, abs/1111.5305, 2011.

\end{thebibliography}


\end{document}